\documentclass[a4paper,UKenglish,pdfa,thm-restate,numberwithinsect]{article}

\usepackage[utf8]{inputenc}
\usepackage{fullpage}
\usepackage{enumitem}
\usepackage{amsmath, amsthm, amssymb, mathtools, xcolor, framed, tikz}
\usepackage{thmtools}
\usepackage{algorithm}
\usepackage{algpseudocode}
\usetikzlibrary{positioning}
\usetikzlibrary{shapes}
\usetikzlibrary{decorations.pathreplacing}
\usetikzlibrary{calc,patterns,angles,quotes}
\usetikzlibrary{intersections}
\usepackage{hyperref}
\usepackage{float}
\hypersetup{colorlinks=true}
\usepackage[capitalise]{cleveref}
\usepackage{soul}
\usepackage{cite}
\usepackage{mathabx}
\usepackage{regexpatch}

\makeatletter
\@addtoreset{ALG@line}{algorithm}
\makeatother

\newcommand{\Ddt}{\textrm{D}^{dt}}
\newcommand{\Rdt}{\textrm{R}^{dt}}
\newcommand{\DmAdt}{\textrm{D}^{\wedge dt}}
\newcommand{\RmAdt}{\textrm{R}^{\wedge dt}}

\newcommand{\DSize}{\textrm{DSize}^{dt}}
\newcommand{\RSize}{\textrm{RSize}^{dt}}

\newcommand{\s}{\textrm{s}}

\newcommand{\mbs}{\textrm{MBS}}

\newcommand{\spar}{\textrm{spar}}
\newcommand{\gspar}{\textrm{gspar}}
\newcommand{\gdeg}{\textrm{gdeg}}
\newcommand{\aspar}{\widetilde{\textrm{spar}}}
\newcommand{\adeg}{\widetilde{\textrm{deg}}}
\newcommand{\agspar}{\widetilde{\textrm{gspar}}}
\newcommand{\gadeg}{\widetilde{\textrm{gdeg}}}
\newcommand{\OR}{\mbox{{\sc Or}}}
\newcommand{\pOR}{\mbox{{\sc promise-Or}}}
\newcommand{\AND}{\mbox{{\sc And}}}
\newcommand{\XOR}{\mbox{{\sc Xor}}}
\newcommand{\Parity}{\mbox{{\sc Parity}}}
\newcommand{\Thr}{\mbox{{\sc Thr}}}
\renewcommand{\epsilon}{\varepsilon}
\newcommand{\samplerho}{\textsc{MaxDegreeRestriction}}
\newcommand{\samplerhog}{\textsc{MaxSensitivityRestriction}}
\newcommand{\supp}{V}
\newcommand{\vars}{\mathsf{Vars}}
\newcommand{\polylog}{\mathrm{polylog}}
\newcommand{\cF}{\mathcal{F}}
\newcommand{\cP}{\mathcal{P}}
\newcommand{\cN}{\mathcal{N}}
\newcommand{\srho}[1]{|#1|_*}
\newcommand{\dmnorm}[1]{\textrm{wt}({#1})}
\newcommand{\dmnormapprox}[2]{\widetilde{\textrm{wt}}_{#2}({#1})}

\newcommand{\wt}{\textrm{wt}}
\newcommand{\wtapprox}{\widetilde{\textrm{wt}}}

\newcommand{\gdmnorm}[1]{\textrm{gwt}({#1})}
\newcommand{\gdmnormapprox}[1]{\widetilde{\textrm{gwt}}({#1})}

\declaretheorem[name=Theorem,numberwithin=section]{theorem}

\declaretheorem[name=Claim,sibling=theorem]{claim}
\declaretheorem[name=Observation,sibling=theorem]{observation}

\declaretheorem[name=Definition,sibling=theorem]{definition}

\declaretheorem[name=Remark,sibling=theorem]{remark}
\declaretheorem[name=Question,sibling=theorem]{question}

\makeatletter
\xpatchcmd\thmt@restatable{%
  \csname #2\@xa\endcsname\ifx\@nx#1\@nx\else[{#1}]\fi
}{%
  \ifthmt@thisistheone
    \csname #2\@xa\endcsname\ifx\@nx#1\@nx\else[{#1}]\fi
  \else
    \csname #2\@xa\endcsname[{Restated}]
  \fi
}{}{}
\makeatother

\title{Exact versus Approximate Representations of Boolean Functions in the De Morgan Basis}
\author{
Arkadev Chattopadhyay
\thanks{TIFR, Mumbai. Email: {\tt arkadev.c@tifr.res.in}. Supported by the Department of Atomic Energy, Govmt. of India, under project \#RTI4001 and by a Google India Faculty Award.} 
\and 
Yogesh Dahiya 
\thanks{TIFR, Mumbai. Email: {\tt yogesh.dahiya@tifr.res.in}. Supported by the Department of Atomic Energy, Govmt. of India, under project \#RTI4001 and the R. Narasimhan postdoctoral fellowship.} 
\and 
Shachar Lovett 
\thanks{UC San Diego. Email: {\tt slovett@ucsd.edu}. Supported by Simons Investigator Award \#929894 and NSF award CCF-2425349.} 
}
\date{}

\begin{document}
\maketitle
\begin{abstract}
A seminal result of Nisan and Szegedy (STOC, 1992) shows that for any total Boolean function, the degree of the real polynomial that computes the function, and the minimal degree of a real polynomial that point-wise approximates the function, are at most polynomially separated. Extending this result from degree to other complexity measures like sparsity of the polynomial representation, or total weight of the coefficients, remains poorly understood. 

In this work, we consider this problem in the De Morgan basis, and prove an analogous result for the sparsity of the polynomials at a logarithmic scale. Our result further implies that the exact $\ell_1$ norm and its approximate variant are also similarly related to each other at a log scale. This is in contrast to the Fourier basis, where the analog of our results are known to be false. 

Our proof is based on a novel random restriction method. Unlike most existing random restriction methods used in complexity theory, our random restriction process is adaptive and is based on how various complexity measures simplify during the restriction process.
\end{abstract}

\section{Introduction}
Polynomial representations of Boolean functions have been invaluable in theoretical computer science and discrete mathematics. While the representation could use any field, in this work we consider only polynomials over the reals. Two bases are particularly prominent.

The first arises naturally by viewing the domain of Boolean functions as $\{0,1\}^n$, and hence, every multilinear monomial just represents the Boolean $\AND$ of a subset of variables. This is known as the \emph{De Morgan basis}. The other basis comes about by viewing the domain as $\{1,-1\}^n$ which is a simple linear transformation of $\{0,1\}^n$ that maps $0 \mapsto 1$ and $1 \mapsto -1$. In this basis, called the \emph{Fourier basis}, each monomial represents the Boolean parity of a subset of variables. 

Two natural complexity measures show up in either basis: the degree and sparsity of the representation. As every Boolean function has a unique representation in either basis, it is usually quite straightforward to determine the degree and sparsity of the unique representation for a function $f$, which we denote by $\deg(f)$ and $\spar(f)$ in the De Morgan basis, and by $\deg^{\oplus}(f)$, and $\spar^{\oplus}(f)$ in the Fourier basis. The linear invertible mapping from one basis to the other ensures that for every $f$, $\deg(f) = \deg^{\oplus}(f)$. But sparsity can be very sensitive to the basis chosen. For example, the $n$-bit AND function has sparsity 1 in the De Morgan basis and $2^n$ in the Fourier basis; and the $n$-bit Parity function has sparsity $1$ in the Fourier basis and $2^n$ in the De Morgan basis.
We understand reasonably satisfactorily exact polynomial representations of Boolean functions. However, when we turn to approximations, the picture becomes subtler.

Classical approximation theory deals with polynomials that point-wise approximate functions. In an influential work, Nisan and Szegedy \cite{nisan1994degree} introduced this notion to the study of Boolean functions. In particular, they defined the complexity measure of approximate degree of a Boolean function $f$, denoted by $\adeg(f)$, to be the smallest degree needed by a polynomial to point-wise approximate $f$ to within a constant distance that is, by default, taken to be $1/3$. Observe that the same reasoning as applied above for exact degree implies that approximate degree of a function is also a measure that is independent of the basis.  The notion of approximate degree has had tremendous impact in computer science as it is related to many other complexity measures including the randomized and quantum query complexity of $f$~\cite{BuhrmanW02,BealsBCMW01, AaronsonS04, BunKT20,BunT22} and the quantum and classical communication complexity of appropriately lifted functions~\cite{BuhrmanW01, R03, SZ09, Sherstov11, ChattopadhyayA08, LeeS09, BeameH09,Sherstov14}. It has also found applications in learning theory~\cite{klivans2004learning,kalai2008agnostically}, differential privacy \cite{thaler2012faster,chandrasekaran2014faster}, secret sharing~\cite{BIVW16,BMTW19}, and many other areas. Unlike exact degree, getting tight bounds on approximate degree often turns out to be challenging. However, Nisan and Szegedy proved a remarkable structural result that for every total Boolean function $f$, the approximate and exact degree of $f$ are polynomially related to each other. One of the striking applications of this result is the polynomial equivalence of quantum and classical query models for total Boolean functions, first derived by Beals et al. \cite{BealsBCMW01}. In a much more recent work, building upon Huang's breakthrough proof \cite{Huang-19} of the sensitivity conjecture, Aaronson {et al.} \cite{aaronson2021degree} finally gave a tight relationship between the two measures by showing that $\deg(f) = O(\adeg(f))^2$. The tightness is witnessed by the Boolean $\AND$ and $\OR$ functions.

Given Nisan and Szegedy's result, one naturally wonders if approximation could reduce \emph{sparsity} needed for total functions. In the Fourier basis, it is known that approximation does reduce sparsity exponentially. For instance, the Fourier sparsity of the $n$-bit $\AND$ function is $2^n$. However, it can be shown that its Fourier approximate-sparsity is $O(n^2)$ (implicit in Bruck and Smolesnky \cite[appendix]{BS90} and explicit in \cite[Lemma 2.8]{CMS20}).  Surprisingly, the question if approximation helps significantly in the De Morgan basis remained unaddressed.  Very recently, Knop {et al.} \cite{knop2021guest} conjectured that in the De Morgan basis, approximation should not significantly reduce sparsity for any total Boolean function. Our main result, stated below, confirms this conjecture. We denote by $\aspar(f)$, the approximate-sparsity of $f$ in the De Morgan basis.

\begin{restatable}[Main Theorem]{theorem}{thmsparaspar}\label{thm:spar-aspar}
For every total Boolean function $f : \{0,1\}^n \to \{0,1\}$, we have
\[
    \log(\spar(f)) = O\big(\log(\aspar(f))^2 \cdot \log n\big).
\]
\end{restatable}

Before we continue, let us remark on the tightness of this result.

\begin{remark} \label{remark:main}
The $n$-bit $\OR$ function has sparsity $2^n - 1$ and approximate sparsity $2^{O(\sqrt{n} \log n)}$, showing that exponential gaps may exist between the two measures in the absolute scale. 
It is thus necessary to consider the log scale as done in Theorem~\ref{thm:spar-aspar} for seeking polynomial relationship between the two measures. The example of $\OR$ also demonstrates the tightness of our bound up to poly-logarithmic factors. Finally, the appearance of the ambient dimension $n$ in our result is unavoidable. Consider the function \( \Thr^n_{n-1} : \{0,1\}^n \to \{0,1\} \) defined by
\[
\Thr^n_{n-1}(x) = 1 \quad \text{if and only if} \quad |x| \geq n-1,
\]
namely, the function evaluates to 1 if the input has at most one zero. It's simple to verify that its exact sparsity is \( n + 1 \), %
and we show in section~\ref{subsec:spar-optimality}  that its approximate sparsity is \( O(\log n) \), implying that an additive \( O(\log n) \) or multiplicative \( O\left(\frac{\log n}{\log \log n}\right) \) factor is necessary in \cref{thm:spar-aspar}.

\end{remark}

One of the motivations of Nisan and Szegedy to study approximate degree was to relate this measure with decision tree complexity. Let $\Ddt(f)$ and $\Rdt(f)$ denote the deterministic and randomized bounded-error decision tree complexities of $f$ respectively. Their result, along with the recent improvement of \cite{aaronson2021degree} yields the following relationship.

\begin{theorem}[Nisan-Szegedy {\cite[Theorem 1.5]{nisan1994degree}}  + Aaronson {et al.} \cite{aaronson2021degree}]  \label{thm:NSAaronson}

For every total Boolean function $f$, the following holds:

$$ \adeg(f) \le \Rdt(f) \le \Ddt(f) \le O(\adeg(f)^4).$$

\end{theorem}

Just as $\adeg(f)$ lower bounds $\Rdt(f)$, it is straightforward to verify that $\log(\aspar(f))$, up to an additive \( \log n \) term, lower bounds the randomized $\AND$-decision tree (ADT) complexity of $f$. In an ADT, each internal node queries the $\AND$ of a subset of variables. ADT's have connections to combinatorial group testing algorithms and have also been the subject of several recent works \cite{Nisan21,BN21,knop2021log,chattopadhyay2023randomized}. We denote the deterministic and randomized ADT complexities of $f$ by $\DmAdt(f)$ and $\RmAdt(f)$ respectively. 

Combining our main result with the recent result of Knop {et al.} \cite{knop2021log} yields the following ADT analog of Theorem~\ref{thm:NSAaronson}.

\begin{restatable}{theorem}{sparsityADT}\label{thm:sparsity-ADT} 
For every total Boolean function \( f : \{0,1\}^n \to \{0,1\} \), the following holds:
\[
\Omega(\log(\aspar(f)) - \log n) \overset{(1)}{=} \RmAdt(f) \leq \DmAdt(f) \overset{(2)}{=} O\left((\log \aspar(f))^{6} \cdot \log n\right).
\]
\end{restatable}

It is worth noting that one could go one step further in the chain of inequalities to show that $\DmAdt(f)$ is upper bounded by $O(\RmAdt(f)^{6})$ upto $\polylog$  factors, thus concluding that randomization doesn't yield more than polynomial savings over the cost of deterministic ADT algorithms. Such a conclusion, in fact with a better polynomial bound, was first derived recently by Chattopadhyay, Dahiya, Mande, Radhakrishnan and Sanyal \cite{chattopadhyay2023randomized}. But our current technique is quite different and the previous result could not give an upper bound on ADT complexity in terms of approximate sparsity as we do here. 

Apart from degree and sparsity, there is a third complexity measure that has been well investigated in the Fourier basis. This is the Fourier $\ell_1$ norm, also called the spectral norm of a Boolean function $f$. Denoted by $\|\hat{f}\|_1$, it is defined as  the sum of the magnitude of the Fourier coefficients of $f$. It appeared in the context of additive combinatorics~\cite{GS08}, communication complexity of $\XOR$ functions~\cite{CM17eccc,CMS20,cheung2025lower} and analysis of Boolean functions~\cite{BS90,AFH12,CM17}. One naturally defines its $\epsilon$-approximate version, denoted by $\|\hat{f}\|_{1,\epsilon}$, to be the amount of Fourier $\ell_1$ mass needed by any real-valued function $g$ to point-wise approximate $f$ within $\epsilon$. Can approximation reduce significantly the needed $\ell_1$ mass? Very recently, Cheung, Hatami, Hosseini, Nikolov, Pitassi and Shirley \cite{cheung2025lower}, constructed a Boolean function $f$ such that $\log(\|\hat{f}\|_{1,1/3})$ is exponentially smaller than $\log(\|\hat{f}\|_1)$, which implies that the Fourier basis yields exponential advantage to approximation even with respect to the spectral norm.

In contrast, the proof method that we develop to establish Theorem~\ref{thm:spar-aspar} shows that approximation does not significantly reduce even the $\ell_1$ mass of a total Boolean function in the De Morgan basis. More precisely, let $\dmnorm{f}$ and $\dmnormapprox{f}{\epsilon}$ represent the exact and $\epsilon$-approximate $\ell_1$ norm of $f$ in the De Morgan basis (we write $\dmnormapprox{f}{} := \dmnormapprox{f}{1/3}$ when $\epsilon = 1/3$).

\begin{restatable}{theorem}{thmellnorm}\label{thm:ell-1} 
For every total Boolean function $f:\{0,1\}^n \to \{0,1\}$, we have
\[
\log \dmnorm{f} = O\left((\log \dmnormapprox{f}{})^2 \cdot \log n\right).
\]
\end{restatable}

\paragraph{General representations:} A natural question emerges from our results. Let $\mathcal{F}$ be a family of elementary real-valued functions defined over the $n$-ary Boolean cube $B_n$, such that $\mathcal{F}$ spans the vector space $\mathbb{R}^{B_n}$ of all real-valued functions. The \emph{sparsity} (resp., \emph{weight}) of a (Boolean) function $f$ wrt $\cF$ is defined as the minimum integer (resp., non-negative real number) $k \ge 0$ such that $f$ can be written as a linear combination of at most $k$ functions from $\cF$ (resp., with total absolute coefficient sum at most $k$)\footnote{Note that there may be more than one way of doing that.}. Denote these complexity measures by $\spar_{\cF}(f)$ and $\wt_{\cF}(f)$ respectively. For example, when $\cF$ is the family of all $\AND$ functions, these measures correspond to the De Morgan sparsity and $\ell_1$ norm of $f$, and when $\cF$ corresponds to all parities, these correspond to the Fourier sparsity and the spectral norm of $f$.
Likewise, one defines the approximate sparsity and weight of $f$ with respect to $\cF$, denoting them  by $\aspar_{\cF}(f)$ and $\wtapprox_{\cF}(f)$ respectively. 

\begin{question}  \label{general-sparsity}

What properties of $\cF$ ensure that approximation doesn't help reduce sparsity or the weight of a Boolean function, i.e. do there exist constants $\alpha$ and $\beta$ such that for all Boolean functions $f$, $\log(\spar_{\cF}(f))  = O((\log(\aspar_{\cF}(f))^{\alpha})$ and/or $\log(\wt_{\cF}(f))  = O((\log(\wtapprox_{\cF}(f))^{\beta})$?

\end{question}

This question is quite broad. For instance, if one views the input domain of the functions as the set of $m \times m$ Boolean matrices, and $\cF$ be the set of all rank one matrices, then Question~\ref{general-sparsity} specializes to asking if log of the rank of a Boolean matrix is always at most a fixed polynomial of the log of its approximate-rank. In general, this is false. For example, the identity matrix has rank  $m$ but its approximate rank is $(\log m)^{O(1)}$. However, it is unknown for special classes of Boolean matrices like those that are the truth table of \(\AND\)-functions (i.e., functions composed with 2-bit \(\AND\) gadgets). Understanding the power of approximation for such special classes of matrices is of significant interest, given its connection to quantum communication complexity. We talk more about this aspect in Section~\ref{sec:related}. 

Summarizing what we have seen, if $\cF$ is the set of all parities, i.e. the Fourier monomials, then approximation can significantly help, and reduce sparsity exponentially. We showed that if $\cF$ is the set of all monotone Boolean $\AND$ functions, i.e. the De Morgan basis, then approximations do not help, and reduce sparsity by at most a polynomial factor on the log scale. In fact our main result gives us slightly more: let $[n]$ be partitioned into two sets, the set of positive literals denoted by $\cP$ and the set of negated literals $\cN$. Each such partition defines a \emph{shifted} De Morgan basis, where a shifted monomial is given by a pair of sets $P\subseteq \cP$ and $N \subseteq \cN$, and corresponds to the Boolean function $M_{P,N} \coloneqq \prod_{i \in P} x_i \prod_{j \in N} (1-x_j)$. Observe that while $\OR$ has full De Morgan sparsity, it has sparsity just 2 in the fully shifted De Morgan basis, i.e., the basis that corresponds to $\cP = \emptyset$ ans $\cN = [n]$. There are $2^n$ such shifted bases, and it is straightforward to verify that our main results—Theorem~\ref{thm:spar-aspar} and Theorem~\ref{thm:ell-1}—imply that, in each shifted basis, the approximate sparsity and approximate \( \ell_1 \)-norm are polynomially related to the exact sparsity and exact \( \ell_1 \)-norm. A natural generalization of the case when $\cF$ is just a shifted De Morgan basis, is the case when we populate the set $\cF$ with all shifted monomials. More precisely, consider $\cF \coloneqq \{M_{P,N} \,:\, P,N \subseteq [n],\, P \cap N = \emptyset\}$, where each $M_{P,N} \coloneqq \prod_{i \in P} x_i \prod_{j \in N} (1-x_j)$ is called a \emph{generalized monomial}. Observe that any shifted De Morgan basis is a strict subset of $\cF$, the set of generalized monomials, whose size is $3^n$. The following concrete question, which is a special case of Question~\ref{general-sparsity}, remains intriguingly open!
\begin{question} \label{question-gen-monomial}
Does there exist a total Boolean function $f$ for which the approximate generalized-monomial sparsity (approximate generalized-monomial weight), denoted by $\agspar(f)$ ($\gdmnormapprox{f}$), is super-polynomially smaller in the log scale than its exact generalized-monomial sparsity (generalized-monomial weight), denoted by $\gspar(f)$ ($\gdmnorm{f}$)?
\end{question}

As expected, generalized monomials can significantly reduce sparsity compared to any shifted De Morgan basis. For instance, consider the following function that mixes two shifted OR's by a monotone addressing scheme: let $f_{\text{mixed}} : \{0,1\}^2 \times \{0,1\}^n \to \{0,1\}$, where $f_{\text{mixed}}(x,y)$ outputs $0$ if $x = 00$, outputs $1$ if $x = 11$, computes the Boolean $\OR$ of $y$ if $x = 10$, and computes the Boolean $\AND$ of $y$ if $x = 01$. It is easy to verify that $\gspar(f_{\text{mixed}}) = O(1)$, while the sparsity of $f_{\text{mixed}}$ in any shifted De Morgan basis is $2^{\Omega(n)}$. Observing that $f_{\text{mixed}}$ is a monotone function, it becomes interesting to answer Question~\ref{question-gen-monomial} for monotone functions. We provide a negative answer below.

\begin{restatable}{theorem}{thmgspar} \label{thm:mon-gen-sparsity}
For every monotone Boolean function \( f : \{0,1\}^n \to \{0,1\} \), the following hold:
\begin{enumerate}
    \item[(a)] \( \log \gspar(f) = O\left((\log \agspar(f))^4 \cdot (\log n)^3\right) \),
    \item[(b)] \( \log \gdmnorm{f} = O\left((\log \gdmnormapprox{f})^4 \cdot (\log n)^3\right) \).
\end{enumerate}
\end{restatable}

The above result, in fact, yields the following more detailed picture about query complexity. Recall that the size of a decision tree is the number of leaves in it. The deterministic (randomized) decision tree size complexity of a function $f$, denoted by $\DSize(f)$ ($\RSize(f)$), is the smallest size needed by an ordinary deterministic (randomized) decision tree to compute $f$.  We can use Theorem~\ref{thm:mon-gen-sparsity} and other standard results to get the following.

\begin{restatable}{corollary}{corrgenspar}  
\label{corr-monotone}
For every monotone Boolean function \( f:\{0,1\}^n \to \{0,1\} \), the following hold:
\begin{enumerate}
    \item [(a)]
    \(
    \Omega(\log(\agspar(f)) - \log n) \overset{(1)}{=} \log \RSize(f) \leq \log \DSize(f) \overset{(2)}{=} O\left((\log \agspar(f))^4 \cdot (\log n)^3\right).
    \)
    
    \item [(b)]
    \(
    \Omega(\log(\gdmnormapprox{f})) \overset{(1)}{=} \log \RSize(f) \leq \log \DSize(f) \overset{(2)}{=} O\left((\log \gdmnormapprox{f})^4 \cdot (\log n)^3\right).
    \)
\end{enumerate}
\end{restatable}

\begin{remark}\label{rmk:sinkgspar}
In particular, this yields the fact that deterministic and randomized decision tree size measure of a monotone function are, upto poly-log$(n)$ factors, polynomially related in the log scale. Such a relationship was recently proven to be true even for general functions in Chattopadhyay et. al. \cite{chattopadhyay2023randomized}. However, the tighter relationship that we prove here via approximate generalized sparsity and weight for monotone functions, is known to be false for general functions as witnessed by the Sink function\footnote{Sink was used to construct a counter-example to the Log-Approximate-Rank Conjecture in \cite{CMS20}}. The Sink function has ${n \choose 2}$ input bits, corresponding to the edges of a complete graph on $n$ vertices. Each Boolean assignment orients the edges. Sink outputs 1 iff there exists a Sink vertex in the resulting directed graph. Its generalized sparsity and weight is just $n$, whereas $\RSize\big(\text{Sink}\big)$ is $2^{\Omega(n)}$.
\end{remark}

\subsection{Our Method} \label{sec:method}

Lower bounds on approximate sparsity were known for specific functions such as \( \OR_n \) and \( \Parity_n \) (these are folklore results), typically established using random restrictions and approximate degree lower bounds. However, these results apply only to specific functions or restricted classes of functions. The general idea is to apply a random restriction \( \rho \), which selects a random subset of variables and fixes each to 0, with the goal of eliminating high-degree monomials from a candidate sparse polynomial \( P \) that approximates \( f \), such that \( f|_{\rho} \) still has large approximate degree while \( P|_{\rho} \) has degree that is too small, yielding a contradiction. This is illustrated by considering the \( n \)-bit \( \OR \) function. Let \( P \) be any sparse polynomial approximating \( \OR_n \). Consider a random restriction \( \rho \) that, independently for each of the \( n \) variables, fixes it to 0 with probability \( 1/2 \) and leaves it free with probability \( 1/2 \). With high probability, at least \( n/3 \) variables are left free. On the other hand, any monomial of degree larger than \( \sqrt{n} \) survives (i.e., none of its variables are set to 0) with probability at most \( 2^{-\sqrt{n}} \). If the number of monomials in \( P \) is \( s \), then the probability that \( P|_{\rho} \) contains a monomial of degree larger than \( \sqrt{n} \) is at most \( s \cdot 2^{-\sqrt{n}} \), which is less than \( 1/2 \) if \( s < 2^{\sqrt{n}}/2 \). With high probability, \( (\OR_n)|_{\rho} \) is an \( r \)-bit \( \OR \) function with \( r \geq n/3 \). Since \( \OR_n|_{\rho} \) is still approximated by \( P|_{\rho} \) (for every \( \rho \)), and recalling that the approximate degree of \( \OR_r \) is \( \Omega(\sqrt{r}) \), we conclude that \( s = 2^{\Omega(\sqrt{n})} \).

While this works for $\OR$ function, there are functions which are very different from $\OR$ and yet have large sparsity in De Morgan basis. An example of that is $\AND_n \circ \OR_2$, where the bottom ORs are 2-bit functions. It is simple to verify that this function has sparsity $2^{\Omega(n)}$. But there is no way to induce a large $\OR$ in this function. If one tried to apply a random restriction like the one that worked for $\OR$, one concludes easily that it won't work as with high probability one of the bottom $\OR_2$ will have both its input variables fixed to 0, thereby killing the entire function. One way to fix this is to consider a slightly more careful restriction. For each of the bottom ORs, one selects one of its two input variables at random and fixes it to 0, leaving the other variable free. It is not hard to show that in this case the restricted function is always the $\AND$ over the remaining \(n\) free variables, and if the approximating polynomial for the $\AND_n \circ \OR_2$ had sparsity \(2^{o(\sqrt{n})}\), then with nonzero probability, the restricted polynomial would give an \(o(\sqrt{n})\)-degree approximation to \(\AND_n\), contradicting known lower bounds. The important thing to note here is that our random restriction is no longer done independently for each variable, as the restrictions on the two variables in each $\OR_2$ block are correlated. 

Our approach generalizes this idea. We design a random restriction process that works for all functions with large exact sparsity. It's not a-priori clear what useful combinatorial structural information can be extracted from just knowing that a function has large sparsity. This is precisely the main technical contribution of our work. We devise a random restriction procedure for such functions with large sparsity and it differs from the two simple cases that we considered above in the following two ways: (i) the method is adaptive in the sense that the next bit that is fixed or left free depends on what happened in the previous step\footnote{The restriction that we used for $\AND_n \circ \OR_2$ is still non-adaptive as the fixing in each of the $\OR_2$ block is independent of the other blocks and can therefore be done all at once.} (ii) variables are not exclusively fixed to 0, and some may get fixed to 1 as well. 

As some of the variables may get fixed to 1, we can no longer argue that all high degree monomials of the approximating polynomial are `zeroed' out as was happening for the two cases discussed earlier. Instead, we argue that if the unrestricted approximating polynomial was sparse to begin with, then with high probability, the restricted polynomial will not have any monomial that has large degree in the free variables. Note that a non-zeroed out monomial could either completely collapse to a constant value, in which case it has degree zero with respect to the free variables, or otherwise it could be nonzero but low degree on the free variables. 

It is known, via a result generally attributed to Grolmusz \cite{grolmusz1997power}, that lower bounds on approximate sparsity yields lower bounds on approximate weight ($\ell_1$ norm) in any basis and even wrt any arbitrary set $\cF$ of basic functions. Hence, our lower bound on the sparsity of approximating polynomials when their exact sparsity is large, in fact, yields a lower bound on the weight of an approximating polynomial as well. To relate exact $\ell_1$ norm and approximate $\ell_1$ norm, we need to assume that the exact $\ell_1$ norm, instead of exact sparsity, of $f$ is large. This is indeed a weaker assumption as in the De Morgan basis, all non-zero coefficients in the unique polynomial representation of $f$ are integer-valued. This is precisely what our random restriction procedure \samplerho, described in Section~\ref{subsec:restriction-prop}, does by taking a bit more care. 

We are able to adapt our restriction technique to deal with generalized polynomials as long as the function has an additional property: it has what we call a \emph{separating set} of inputs. A separating set of inputs for $f$ is a special \emph{fooling} set with respect to subcubes. Given such a separating set, we're able to modify our random restriction technique to reduce the degree of generalized monomials with respect to free variables while ensuring that the restricted function retains high degree. Finally, we observe that monotone functions with large generalized sparsity have such large separating sets and this helps us prove that generalized sparsity (generalized weight) and approximate generalized sparsity (approximate generalized weight) are polynomially related in the log scale for monotone functions.  

\subsection{Other Related Work}  \label{sec:related}
Random restrictions have been used to obtain lower bounds for (approximate) sparsity before, say for the $\OR_n$ function (folklore) or for specific functions~\cite{CM17,ACW25}. Similarly, random restrictions have been famously quite successful in other parts of complexity theory like circuit complexity~\cite{FurstSS84,Hastad86} and proof complexity~\cite{BellantoniPU92,PitassiBI93,KrajicekPW95}. Designing new random restrictions in these two areas remains an active, technically challenging theme of current research~\cite{HastadRST17,Hastad23}. In these works, random restrictions are applied to a target functions (in circuits) or CNF formulas (in proof complexity) that have explicit convenient combinatorial properties. 
On the other hand, we design a generic random restriction scheme that can be applied on any function that has the algebraic property of large sparsity in the De Morgan basis. Using this new scheme we derive a general structural result applicable to \emph{all} functions. Such an application of random restriction seems rare to us. It vaguely reminds us of two results, both in proof complexity, where \emph{deterministic} greedy restrictions were used to obtain results that are applicable to all formulas: the first is by Impagliazzo, Pudlak and Sgall \cite{IPS99} showing that degree lower bounds are enough to prove size (read sparsity) lower bounds for proofs in polynomial calculus. The second is by Ben-Sasson and Wigderson \cite{BW01} who also used a greedy restriction method to show that width of resolution proofs always translate to size of proofs. 

 Buhrman and de Wolf \cite{BuhrmanW01} were interested, among other things, in characterizing the bounded-error quantum communication complexity of $\AND$ functions, i.e 2-party communication functions of the form $ f_n \circ \AND_2$, where $f_n$ is an arbitrary $n$-bit Boolean function. This problem has remained open. There was a breakthrough made by Razborov \cite{R03} who showed that for symmetric $f_n$, quantum protocols offer no advantage over classical randomized protocols. Whether there exists some $f_n$ for which the quantum and classical randomized communication complexities are widely separated for $f_n \circ \AND_2$ remains open despite several efforts~\cite{BuhrmanW01,Klauck07,SZ09,Sherstov10}. Our main result provides a conditional answer to that question as discussed below.
 
 Razborov proved his lower bound by showing a lower bound on the log of the approximate rank of the communication matrix. Buhrman and de Wolf \cite{BuhrmanW01}, observing that the exact sparsity of any function $f$ in the De Morgan basis is equal to the rank of the communication matrix of $f \circ \AND_2$, raised informally the following interesting question:

 \begin{question}  \label{q:Buhrman-Wolf}
 Is it true that for every function $f$, (logarithm of) the approximate sparsity of $f$ and (logarithm of) the approximate rank of the communication matrix of $f \circ \AND_2$ are within a polynomial of each other?
 \end{question}

 \begin{remark}
 Buhrman and de Wolf didn't quite phrase this question in the manner we do. Their discussion didn't put any quantitative bounds, nor do they talk about relating things in the log scale. Our version, therefore, may be a significant weakening of what they had in mind.
 \end{remark}

  Our result provides a fresh impetus to seek an answer to the above question for the following reason. Assuming the answer is positive, Theorem~\ref{thm:spar-aspar}, combined with the recent resolution (up to a \( \log n \) factor) of the log-rank conjecture for \( \AND \)-functions by Knop, Lovett, McGuire, and Yuan~\cite{knop2021log}, implies that the deterministic classical zero error communication complexity of every function $f \circ \AND_2$ is at most a fixed polynomial of its quantum bounded-error communication complexity, ignoring poly-logarithmic factors.

\section{Preliminaries}

In this section, we collect notation, definitions, and known results that will be used throughout the paper. All functions considered are defined on the Boolean hypercube \( \{0,1\}^n \), and all polynomials are assumed to be multilinear real polynomials.

\begin{definition}[Multilinear Polynomial Representation]
A polynomial \( Q \in \mathbb{R}[x_1, x_2, \dots, x_n] \) is called \emph{multilinear} if each variable appears with degree at most one in every monomial. Over the Boolean domain, every function \( f : \{0,1\}^n \to \mathbb{R} \) admits a unique multilinear polynomial representation. That is, there exists a unique multilinear polynomial \( Q \in \mathbb{R}[x_1, x_2, \dots, x_n] \) such that \( Q(x) = f(x) \) for all \( x \in \{0,1\}^n \).
\end{definition}

\begin{definition}[Support, Degree, Sparsity, and Norm of a Polynomial]
Let \( Q \in \mathbb{R}[x_1, \dots, x_n] \) be a multilinear polynomial written as
\[
    Q(x) = \sum_{S \subseteq [n]} a_S \prod_{i \in S} x_i.
\]
\begin{itemize}
    \item The \emph{support} of \( Q \), denoted \( \vars(Q) \), is the set of variables that appear in some monomial with a nonzero coefficient.
    \item The \emph{degree} of \( Q \), denoted \( \deg(Q) \), is \( \max\{ |S| \mid a_S \neq 0 \} \).
    \item The \emph{sparsity} of \( Q \), denoted \( \spar(Q) \), is the number of nonzero coefficients \( a_S \).
    \item The \( \ell_1 \)-norm of \( Q \), denoted \( \dmnorm{Q} \), is given by \( \sum_{S \subseteq [n]} |a_S| \).
\end{itemize}
\end{definition}

\begin{definition}[Complexity Measures for Functions via Polynomials]
Let \( f : \{0,1\}^n \to \mathbb{R} \) be a function, and let \( \mathcal{P}(f) \) denote its unique multilinear polynomial representation. We define the following complexity measures:
\[
    \deg(f) := \deg(\mathcal{P}(f)), \quad
    \spar(f) := \spar(\mathcal{P}(f)), \quad
    \dmnorm{f} := \dmnorm{\mathcal{P}(f)}.
\]
\end{definition}

\begin{remark}\label{rmk:spar-l1}
For Boolean functions \( f : \{0,1\}^n \to \{0,1\} \), the coefficients in \( \mathcal{P}(f) \) are integers (see for example \cite[Chapter 2]{Jukna-BFCbook}), and hence \( \spar(f) \leq \dmnorm{f} \).
\end{remark}

\begin{definition}[Complexity Measures for Functions via Approximating Polynomials]
Let \( f : \{0,1\}^n \to \mathbb{R} \) and let \( \epsilon > 0 \). We define:
\begin{align*}
    \adeg_\epsilon(f) &:= \min \{ \deg(Q) \mid Q \text{ satisfies } |Q(x) - f(x)| \leq \epsilon \text{ for all } x \in \{0,1\}^n \}, \\
    \aspar_\epsilon(f) &:= \min \{ \spar(Q) \mid Q \text{ satisfies } |Q(x) - f(x)| \leq \epsilon \text{ for all } x \in \{0,1\}^n \}, \\
    \dmnormapprox{f}{\epsilon} &:= \min \{ \dmnorm{Q} \mid Q \text{ satisfies } |Q(x) - f(x)| \leq \epsilon \text{ for all } x \in \{0,1\}^n \}.
\end{align*}
When \( \epsilon = 1/3 \), we write \( \adeg(f) := \adeg_{1/3}(f) \), \( \aspar(f) := \aspar_{1/3}(f) \), and \( \dmnormapprox{f}{} := \dmnormapprox{f}{1/3} \).

\end{definition}

\begin{theorem}[{\cite[Claim 4.3]{diakonikolas2010bounded}}]
\label{thm:error-reduction}
Let \( f : \{0,1\}^n \to \{0,1\} \) be a Boolean function. Then for any \( 0 < \epsilon < \frac{1}{2} \),
\[
    \adeg_\epsilon(f) = O\left(\adeg_{1/3}(f) \cdot \log(1/\epsilon)\right).
\]
\end{theorem}

\begin{theorem}[{\cite[Theorem 4]{aaronson2021degree}}]
\label{thm:degree-apxdegree}
For every Boolean function \( f : \{0,1\}^n \to \{0,1\} \),
\[
    \deg(f) = O(\adeg(f)^2).
\]
\end{theorem}

\begin{remark}
The bound in Theorem~\ref{thm:degree-apxdegree} is tight; for example, the \( \OR_n \) function satisfies \( \deg(\OR_n) = n \) and \(\adeg(\OR_n) = \Theta(\sqrt{n}) \).
\end{remark}

\begin{definition}[Restrictions]
A \emph{restriction} \( \rho \) on a set of variables \( V \subseteq \{x_1, \dots, x_n\} \) is a partial assignment
\[
    \rho : V \to \{0,1,*\},
\]
where for \( x_i \in V \), \( \rho(x_i) \in \{0,1\} \) indicates that \( x_i \) is fixed, and \( \rho(x_i) = * \) means \( x_i \) is left free. We define:
\begin{align*}
    \mathsf{SetVars}(\rho) &:= \{ x_i \in V \mid \rho(x_i) \in \{0,1\} \}, \\
    \mathsf{FreeVars}(\rho) &:= \{ x_i \in V \mid \rho(x_i) = * \}.
\end{align*}
The \emph{size} of \( \rho \), denoted \( \srho{\rho} \), is the number of free variables:
\[
    \srho{\rho} := |\mathsf{FreeVars}(\rho)|.
\]
\end{definition}

\noindent
Let \( Q \in \mathbb{R}[x_1, \dots, x_n] \) be a polynomial and \( V \subseteq \{x_1, \dots, x_n\} \). For a restriction \( \rho \) on \( V \), we write \( Q|_\rho \) for the polynomial obtained by substituting \( x_i = \rho(x_i) \) for all \( x_i \in \mathsf{SetVars}(\rho) \).

\medskip

\noindent
For an input \( w \in \{0,1\}^V \) and a subset \( T \subseteq V \), we write \( w|_T \in \{0,1\}^T \) to denote the projection of \( w \) to the coordinates in \( T \). For singleton sets, we simply write \( w_i \) for \( w|_{\{x_i\}} \).

\medskip

\noindent
Let \( F \subseteq \{0,1\}^V \) be a set of Boolean assignments and \( \rho \) a restriction on \( V \). The restriction of \( F \) under \( \rho \), denoted \( F|_\rho \), is defined as
\[
    F|_\rho := \left\{ w|_{\mathsf{FreeVars}(\rho)} \,\middle|\, w \in F,\ \forall x_i \in \mathsf{SetVars}(\rho),\ w_i = \rho(x_i) \right\}.
\]

\section{Sparsity vs. Approximate Sparsity}

In this section, we show that for Boolean functions \( f : \{0,1\}^n \to \{0,1\} \), the exact and approximate sparsity are polynomially related on the logarithmic scale. Rather than assuming large sparsity and arguing for large approximate sparsity, we start with the weaker assumption of a large exact \( \ell_1 \)-norm (see \cref{rmk:spar-l1}) and argue large approximate sparsity. This approach yields both \cref{thm:spar-aspar} and \cref{thm:ell-1} in one go, showing that the logarithms of the exact and approximate sparsity, as well as of the exact and approximate \( \ell_1 \)-norm, are polynomially related up to a \( \log n \) factor.

\paragraph{Proof Overview.}
Let \( f \) be a Boolean function with large exact \( \ell_1 \)-norm. We aim to show that any polynomial approximating \( f \) within error \( 1/3 \) must also has large sparsity. The argument proceeds via a carefully constructed random restriction \( \rho \), sampled using Algorithm~\ref{algo:restriction}, which satisfies the following properties:

\begin{enumerate}
    \item With high probability, \( \rho \) leaves \( \ell = \Omega(\log \dmnorm{f}/\log n) \) variables free.
     \item The restricted function \( f|_{\rho} \) has full degree on the variables left free.
    \item For any monomial \( M \) and any $t \ge 1$, the probability that \( \deg(M|_\rho) \geq t \) is at most \( 2^{-t} \).
\end{enumerate}

With the above properties of $\rho$, the reason why the approximate sparsity of \( f \) must be large becomes evident. Let $\ell$ (roughly $\log \dmnorm{f}/\log n$) denote the number of variables left free by $\rho$. Suppose there exists a polynomial \( Q \) approximating \( f \) having sparsity less than \( 2^{\sqrt{\ell / c}} \), for some constant \( c > 0 \) (to be chosen appropriately). Then, by property (3) above and using a probabilistic argument, there exists a restriction \( \rho \) that eliminates all monomials of degree at least \( \sqrt{\ell / c} \) in \( Q \). Consequently, the restricted polynomial \( Q|_{\rho} \) has degree strictly less than \( \sqrt{\ell / c} \).

On the other hand, by property (2), the restricted function \( f|_{\rho} \) has degree \( \ell \). Therefore, \( Q|_{\rho} \) approximates \( f|_{\rho} \), a Boolean function of degree \( \ell \), using a polynomial of degree less than \( \sqrt{\ell / c} \). This contradicts the known relationship between degree and approximate degree for Boolean functions—specifically, that \( \deg(f) \leq c \cdot \adeg(f)^2 \) for some universal constant \( c \) \cite{aaronson2021degree}.

We conclude that any polynomial approximating \( f \) must have sparsity at least \( 2^{\Omega(\sqrt{\ell})} \). Since \( \ell = \Theta(\log \dmnorm{f}/\log n) \), it follows that the logarithms of exact \( \ell_1 \)-norm  and approximate sparsity are related quadratically (up to a \( \log n \) factor).

The novelty of our proof lies in the method of sampling random restrictions that satisfy the key properties outlined above. In contrast to the non-adaptive restrictions commonly used in circuit complexity and related areas, our sampling procedure is adaptive—it takes into account the effects of previous random choices on the hardness measure (in our case, the sparsity of the restricted function). We believe this adaptive approach to sampling restrictions may have applications beyond the present context.

We now abstract the above idea into a general notion of hardness:

\begin{definition}[\( \ell \)-Variable Max-Degree Distribution]\label{defn:harddist}
Let \( f : \{0,1\}^n \to \{0,1\} \), and let \( \mathcal{D} \) be a distribution over restrictions \( \rho : \{x_1, \dots, x_n\} \to \{0,1,*\} \). We say that \( \mathcal{D} \) is an \emph{\( \ell \)-variable max-degree distribution} for \( f \) if:
\begin{enumerate}
    \item With probability at least \( 0.9 \), \( \rho \) leaves at least \( \ell \) variables free.
    \item For every \( \rho \) in the support of \( \mathcal{D} \), we have \( \deg(f|_\rho) = \srho{\rho} \). 
    \item For any monomial \( M \) and any \( t \in \mathbb{N} \), \(\Pr_{\rho \sim \mathcal{D}}\left[ \deg(M|_\rho) \geq t \right] \leq 2^{-t}.\)
\end{enumerate}
\end{definition}

We will show that a large exact \( \ell_1 \)-norm implies the existence of such a distribution, which in turn implies that any polynomial approximating \( f \)  must have large sparsity—thereby connecting exact \( \ell_1 \)-norm and approximate sparsity.

\paragraph{Organization of this section.}
In~\cref{subsec:restriction-prop}, we show how to construct a max-degree distribution when \( f \) has large exact \( \ell_1 \)-norm. In~\cref{subsec:together}, we use this to prove \cref{thm:spar-aspar,thm:ell-1}. In~\cref{subsec:spar-optimality}, we discuss the tightness of our bounds. Finally, in~\cref{subsec:spar-AND}, we explore implications for the $\AND$ query model.

\subsection{The Restriction Process and Its Properties} \label{subsec:restriction-prop}

\begin{algorithm}[H]
\caption{\samplerho}
\label{algo:restriction}
\begin{algorithmic}[1]
\State \textbf{Input:} Non-zero multilinear polynomial \( Q \in \mathbb{R}[x_1, \dots, x_n] \); set \( \supp \subseteq \{x_1, \dots, x_n\} \) with \( \mathrm{vars}(Q) \subseteq \supp \).
\State \textbf{Output:} A restriction \( \rho: \supp \to \{0,1,*\} \).
\If{\( |\supp| = 0 \)} \label{line:baseCase} 
   \State \Return empty \( \rho \) 
\Else \label{line:nontrivial}
    \If{there exists \( x_i \in \supp \), \( u \in \{0,1\} \) such that \( \dmnorm{Q|_{x_i = u}} \geq \left(1 - \frac{1}{n} \right) \cdot \dmnorm{Q} \)} \label{algo:line:cond2}
        \State \( \rho' \gets \samplerho(Q|_{x_i = u}, \supp \setminus \{x_i\}) \) \label{algo:line:recursivecall1}
        \State Set \( \rho(x_i) \gets u \), and for all \( x_j \in \supp \setminus \{x_i\} \), set \( \rho(x_j) \gets \rho'(x_j) \)
    \Else \label{algo:line:cond3}
        \State Choose \( x_i \in \supp \) arbitrarily \label{algo:line:chosenVar}
        \State Express \( Q \) as \( Q = R_1 \cdot x_i + R_0 \)
        \State With probability \( 1/2 \): \label{algo:line:rand0}
            \State \hspace{1.5em} \( \rho_0 \gets \samplerho(R_0, \supp \setminus \{x_i\}) \) \label{algo:line:recursivecall2}
            \State \hspace{1.5em} Set \( \rho(x_i) \gets 0 \), and for all \( x_j \in \supp \setminus \{x_i\} \), set \( \rho(x_j) \gets \rho_0(x_j) \)
        \State Otherwise: \label{algo:line:rand1}
            \State \hspace{1.5em} \( \rho_* \gets \samplerho(R_1, \supp \setminus \{x_i\}) \) \label{algo:line:recursivecall3}
            \State \hspace{1.5em} Set \( \rho(x_i) \gets * \), and for all \( x_j \in \supp \setminus \{x_i\} \), set \( \rho(x_j) \gets \rho_*(x_j) \)
    \EndIf
    \State \Return \( \rho \)
\EndIf
\end{algorithmic}
\end{algorithm}

Algorithm~\ref{algo:restriction} describes a procedure for sampling random restrictions for a given input polynomial \( Q \). When applied to the unique multilinear polynomial \( Q \) that exactly represents a Boolean function \( f \), we will show that the resulting distribution over restrictions is \( \ell \)-variable max-degree for \( f \), where \( \ell = \Omega\left(\frac{\log \dmnorm{f}}{\log n}\right) \). In this subsection, we establish the properties of the distribution induced by this process.

We begin with some observations about Algorithm~\ref{algo:restriction}. First, we claim that if a call to $\samplerho$ on a polynomial \( Q \) reaches line~\ref{algo:line:cond3}, then \( Q \) is balanced with respect to each of its variables. That is, for every \( x_i \in \vars(Q) \), if we write \( Q = R_1 x_i + R_0 \), then both \( \dmnorm{R_0} \) and \( \dmnorm{R_1} \) are at least a \( (1/2n) \)-fraction of \( \dmnorm{Q} \). Formally:

\begin{claim}\label{claim:balanced}
Suppose Algorithm~\ref{algo:restriction} reaches the \texttt{else} branch at line~\ref{algo:line:cond3} on input polynomial \( Q \). Then for every \( x_i \in \vars(Q) \), writing \( Q = R_1 x_i + R_0 \), we have
\[
\dmnorm{R_0} \geq \frac{1}{2n} \dmnorm{Q} \quad \text{and} \quad \dmnorm{R_1} \geq \frac{1}{2n} \dmnorm{Q}.
\]
\end{claim}

\begin{proof}
If Algorithm~\ref{algo:restriction} reaches the \texttt{else} branch at line~\ref{algo:line:cond3} on input polynomial \( Q \), then by the condition of that line, we have:
\[
\forall x_i \in \vars(Q),\ \forall u \in \{0,1\},\quad \dmnorm{Q|_{x_i = u}} < (1 - \tfrac{1}{n}) \dmnorm{Q}.
\]

Fix a variable \( x_i \in \vars(Q) \), and write \( Q = R_1 x_i + R_0 \). We claim that both \( \dmnorm{R_0} \geq \tfrac{1}{2n} \dmnorm{Q} \) and \( \dmnorm{R_1} \geq \tfrac{1}{2n} \dmnorm{Q} \). Suppose, for contradiction, that one of these inequalities fails.

\begin{itemize}
    \item \textbf{Case 1:} \( \dmnorm{R_1} < \tfrac{1}{2n} \dmnorm{Q} \). Since \( Q|_{x_i = 0} = R_0 \), we get:
    \[
    \dmnorm{Q} = \dmnorm{R_0} + \dmnorm{R_1} = \dmnorm{Q|_{x_i = 0}} + \dmnorm{R_1} < (1 - \tfrac{1}{n}) \dmnorm{Q} + \tfrac{1}{2n} \dmnorm{Q} < \dmnorm{Q},
    \]
    which is a contradiction.

    \item \textbf{Case 2:} \( \dmnorm{R_0} < \tfrac{1}{2n} \dmnorm{Q} \). Since \( Q|_{x_i = 1} = R_1 + R_0 \), cancellations could occur between monomials in \( R_1 \) and \( R_0 \), but even in the worst case we have: \(
    \dmnorm{Q|_{x_i = 1}} \geq \dmnorm{R_1} - \dmnorm{R_0}.\)
    
    Therefore,
    \[
    \dmnorm{Q|_{x_i = 1}} \geq \dmnorm{R_1} - \dmnorm{R_0} = \dmnorm{Q} - 2\dmnorm{R_0} > \left(1 - \tfrac{1}{n}\right)\dmnorm{Q},
    \]
    which contradicts the assumption that \( \dmnorm{Q|_{x_i = 1}} < \left(1 - \tfrac{1}{n} \right) \dmnorm{Q} \).
\end{itemize}

Hence, both \( \dmnorm{R_0} \geq \tfrac{1}{2n} \dmnorm{Q} \) and \( \dmnorm{R_1} \geq \tfrac{1}{2n} \dmnorm{Q} \) must hold for every \( x_i \in \vars(Q) \).
\end{proof}

\begin{observation}\label{obs:non-z}
If the input polynomial \( Q \) is initially non-zero, then by Claim~\ref{claim:balanced}, all recursive calls in Algorithm~\ref{algo:restriction} continue to operate on non-zero polynomials. Furthermore, if at any point during the recursion the sparsity of the input polynomial becomes 1—i.e., the polynomial consists of a single monomial \( M \)—then, due to line~\ref{algo:line:cond2}, the algorithm deterministically sets all remaining variables in the final restriction \( \rho \) according to the support of \( M \). Specifically, for each \( x_i \in \supp \), we set \( \rho(x_i) \gets 1 \) if \( x_i \in \mathrm{vars}(M) \), and \( \rho(x_i) \gets 0 \) otherwise.
\end{observation}

We next observe how the range of the input polynomial over \( \{0,1\}^{\supp} \) evolves during recursion. Let \( Q \) be the input at some stage. The following cases arise in the recursion:
\begin{enumerate}
    \item If the recursion proceeds via line~\ref{algo:line:recursivecall1}, the next polynomial is \( Q|_{x_i = u} \) for some \( x_i \in \supp \), \( u \in \{0,1\} \).
    \item If via line~\ref{algo:line:recursivecall2}, we recurse on \( R_0 = Q|_{x_i = 0} \), where \( x_i \) is chosen in line~\ref{algo:line:chosenVar}.
   \item If via line~\ref{algo:line:recursivecall3}, we write \( Q = R_1 \cdot x_i + R_0 \), where \( x_i \) is chosen in line~\ref{algo:line:chosenVar}, and recurse on \( R_1 = \partial_{x_i} Q \).
\end{enumerate}
In cases (1) and (2), the recursive call uses a restriction of \( Q \), so the range of values on Boolean inputs does not increase. In case (3), the derivative \( R_1 \) may have a larger range. However, as a discrete derivative, its range is controlled—it lies within twice the range of \( Q \).

\begin{claim}\label{claim:rangeQ}
Let \( Q \in \mathbb{R}[x_1, \dots, x_n] \) with \( \vars(Q) \subseteq \supp \), and suppose \( Q = R_1 \cdot x_i + R_0 \) for some \( x_i \in \supp \), such that the range of \( Q \) over \( \{0,1\}^{\supp} \) is contained in \( [a,b] \). Then the range of \( R_1 \) over \( \{0,1\}^{\supp \setminus \{x_i\}} \) is contained in \( [-(b-a), b-a] \).
\end{claim}

\begin{proof}
Fix any \( w \in \{0,1\}^{\supp \setminus \{x_i\}} \), and let \( w_0, w_1 \in \{0,1\}^{\supp} \) be its extensions with \( x_i = 0 \) and \( x_i = 1 \), respectively. Then \( R_1(w) = Q(w_1) - R_0(w) = Q(w_1) - Q(w_0) \). Since both \( Q(w_0) \) and \( Q(w_1) \) lie in \( [a, b] \), their difference lies in \( [-(b - a), b - a] \), as claimed. 
\end{proof}

In particular, if we start with a polynomial \( Q \) representing a Boolean function and take \( k \) successive derivatives, the range of the resulting polynomial is contained in \( [-2^{k-1}, 2^{k-1}] \) by repeated application of the claim.

We now show that the restriction \( \rho \) returned by Algorithm~\ref{algo:restriction} leaves a significant number of variables free. The algorithm proceeds recursively and follows one of three possible execution paths:

\begin{itemize}
    \item If \( |\supp| = 0 \), the recursion terminates and the algorithm backtracks.
    \item If the condition on line~\ref{algo:line:cond2} holds, the algorithm makes a single recursive call (line~\ref{algo:line:recursivecall1}).
    \item Otherwise, the condition on line~\ref{algo:line:cond3} holds, and the algorithm makes one of two recursive calls (lines~\ref{algo:line:recursivecall2}, \ref{algo:line:recursivecall3}) depending on the outcome of a coin toss.
\end{itemize}

The recursion continues while \( |\supp| \geq 1 \), and halts when \( |\supp| = 0 \), after which the final restriction is assembled by backtracking.

We classify recursive calls into two types: a call is \emph{passive} if the condition on line~\ref{algo:line:cond2} is satisfied, and \emph{active} if the condition on line~\ref{algo:line:cond3} is satisfied.

We argue that any execution of the algorithm must involve a substantial number of active recursive calls. In each such call, a variable is left free with probability \( 1/2 \). Therefore, if the algorithm makes \( \ell \) active calls, the expected number of variables left free in the final restriction is \( \ell/2 \). 

Moreover, since the decision to leave a variable free (i.e., assign it the value \(*\)) in an active call is independent of the choices made in previous calls, the number of free variables in the final restriction is tightly concentrated around its expectation. By a standard Chernoff bound, with high probability, at least a constant fraction of these \( \ell \) active calls will indeed result in variables being left free.

This leads to the following formal statement:

\begin{claim}\label{claim:sizeofrho}
Let \( Q \) be the multilinear polynomial representing a Boolean function with \( \dmnorm{Q} \geq 10(4n)^{40} \), and let \( \rho = \samplerho(Q, \{x_1, \dots, x_n\}) \) be the restriction output by Algorithm~\ref{algo:restriction}. Then, with probability at least \( 0.9 \), the restriction \( \rho \) leaves at least \( \Omega\left(\frac{\log \dmnorm{Q}}{\log n}\right) \) variables free.
\end{claim}

\begin{proof}
We begin by showing that any execution of the algorithm must involve a substantial number of \emph{active} recursive calls. Let the algorithm make \( t \) total recursive calls, of which \( \ell \) are active. Since the size of \( \supp \) decreases by 1 in each call and the recursion terminates when \( |\supp| = 0 \), we have \( t \leq n \).

Now, observe how the \( \ell_1 \)-norm and the range of the polynomial (when evaluated on inputs in \( \{0,1\}^{\supp} \)) evolve during the recursion:

\begin{itemize}
    \item In a \emph{passive} call (i.e., when line~\ref{algo:line:cond2} is satisfied), the polynomial in the next step is of the form \( Q|_{x_i = u} \), whose \( \ell_1 \)-norm is at least \( (1 - 1/n) \cdot \dmnorm{Q} \), and which takes the same range of values on Boolean inputs as \( Q \).
    
    \item In an \emph{active} call (i.e., when line~\ref{algo:line:cond3} is satisfied), the \( \ell_1 \)-norm of the next polynomial drops by at most a factor \( 1/(2n) \), by Claim~\ref{claim:balanced}. If recursion proceeds via line~\ref{algo:line:recursivecall2}, the value set over Boolean inputs does not increase. If via line~\ref{algo:line:recursivecall3}, the next polynomial is the discrete derivative \( \partial_{x_i} Q \), whose range on Boolean inputs increases by at most a factor of 2 (Claim~\ref{claim:rangeQ}).
\end{itemize}

The recursion terminates with \( |\supp| = 0 \) and a nonzero constant polynomial. Since each active call can at most double the range of values on Boolean inputs, and the initial polynomial \( Q \) takes values in \( \{0,1\} \), the final constant must lie in \( [-2^\ell, 2^\ell] \). Therefore, the total shrinkage in the \( \ell_1 \)-norm over the course of the recursion satisfies:

\begin{align*}
2^{\ell} &\geq (1 - 1/n)^{t - \ell} \cdot (1/2n)^{\ell} \cdot \dmnorm{Q} \\
  &\geq (1 - 1/n)^n \cdot (1/2n)^{\ell} \cdot \dmnorm{Q}
  && \text{(since \( t - \ell \leq n \))} \\
  &\geq \frac{1}{10} \cdot \left(\frac{1}{2n}\right)^{\ell} \cdot \dmnorm{Q}
  && \text{(using \( (1 - 1/n)^n \geq 1/10 \) for \( n \geq 2 \))}.
\end{align*}
Taking logarithms and rearranging, we obtain:
\[
\ell \geq \frac{\log(\dmnorm{Q}/10)}{\log(4n)}.
\]
Define \( \ell^* := \frac{\log(\dmnorm{Q}/10)}{\log(4n)} \). Thus, every run of the algorithm contains at least \( \ell^* \) active recursive calls.

Let \( X_1, X_2, \dots, X_{\ell^*} \) be indicator random variables, where \( X_i = 1 \) if the variable chosen in the \( i \)-th active call is left free (which occurs with probability \( 1/2 \)), and \( 0 \) otherwise. These variables are independent by construction, and \( \mathbb{E}\left[\sum_{i=1}^{\ell^*} X_i\right] = \frac{\ell^*}{2}.\) By a standard Chernoff bound, we get:
\[
\Pr\left(\sum_{i=1}^{\ell^*} X_i \leq \frac{\ell^*}{4}\right) \leq e^{-\ell^*/16} \leq 0.1,
\]
where the last inequality follows from the assumption \( \dmnorm{Q} \geq 10(4n)^{40} \).

Since the number of variables left free in the final restriction is at least \( \sum_{i=1}^{\ell^*} X_i \), we conclude that with probability at least \( 0.9 \), the algorithm leaves at least \( \ell^*/4 = \Omega\left(\frac{\log \dmnorm{Q}}{\log n}\right) \) variables free.
\end{proof}

Finally, we show that for any input polynomial $Q$, the restriction $\rho$ returned by Algorithm~\ref{algo:restriction} has the following properties: the restricted polynomial \( Q|_\rho \) has full degree—that is, its degree equals the number of variables left free by \( \rho \); and for any monomial \( M \), the degree of \( M|_\rho \) exhibits exponential tail decay: the probability that \( \deg(M|_\rho) \geq t \) is at most \( 2^{-t} \).

\begin{claim}\label{claim:rhoProp}
Let \( Q \in \mathbb{R}[x_1, \dots, x_n] \) be a non-zero polynomial with \( \vars(Q) \subseteq \supp \). Let \( \rho \) be the restriction returned by \( \samplerho(Q, \supp) \), as described in Algorithm~\ref{algo:restriction}. Then:
\begin{enumerate}[label=(\alph*)]
    \item The restricted polynomial \( Q|_\rho \) is non-zero and has full degree; that is, \( \deg(Q|_\rho) = \srho{\rho} \).
    \item For any monomial \( M \) with \( \vars(M) \subseteq \supp \), and any \( t \in \mathbb{N} \), the degree of the restricted monomial satisfies:
    \[
    \Pr_\rho\left(\deg(M|_\rho) \geq t\right) \leq 2^{-t}.
    \]
\end{enumerate}
\end{claim}

\begin{proof}
Claim (b) is trivial when \( t = 0 \), so assume \( t > 0 \). We prove both parts simultaneously by induction on \( |\supp| \).

\textbf{Base Case (\( |\supp| = 0 \)):}
Here, \( Q \) must be a non-zero constant polynomial. The claim holds trivially.

\textbf{Inductive Step (\(|\supp| \geq 1\)):}  
We consider the two possible branches of the algorithm, depending on which condition is satisfied at runtime (line~\ref{algo:line:cond2} or line~\ref{algo:line:cond3}):

\begin{enumerate}

\item \textbf{Case where the ``if'' condition (line~\ref{algo:line:cond2}) is satisfied:}  
Suppose the condition is satisfied for some variable \( x_i \in \supp \) and some value \( u \in \{0,1\} \). The algorithm then returns the restriction \( \rho = \rho' \cup \{x_i \gets u\} \), where \( \rho' \) is obtained from a recursive call with a strictly smaller support set. By the induction hypothesis, the restricted polynomial \( (Q|_{x_i = u})|_{\rho'} \) is non-zero and has full degree. Hence,
\[
\deg(Q|_\rho) = \deg\left((Q|_{x_i = u})|_{\rho'}\right) \overset{(1)}{=} \srho{\rho'} \overset{(2)}{=} \srho{\rho},
\]
where (1) follows from the induction hypothesis, and (2) holds because \( \rho \) and \( \rho' \) leave the same number of variables free. Therefore, \( Q|_\rho \) is non-zero and has full degree.   

Moreover, for any monomial \( M \), we have: 
\[
\Pr_\rho\left(\deg(M|_\rho) \geq t\right) \leq  \Pr_{\rho'}\left(\deg(M|_{\rho'}) \geq t\right) \leq 2^{-t},
\]
where the final inequality follows by the induction hypothesis.

\item \textbf{Case where the else clause at line~\ref{algo:line:cond3} is executed:}  
Let \( x_i \in \supp \) be the variable chosen in line~\ref{algo:line:chosenVar}. Then with probability \(1/2\), the algorithm sets \( x_i \gets 0 \), and returns \( \rho = \rho_0 \cup \{x_i \gets 0\} \); with the remaining probability \(1/2\), it leaves \( x_i \) free (denoted by \( * \)) and returns \( \rho = \rho_* \cup \{x_i \gets *\} \). 

     \begin{itemize}
        \item \textbf{(a) Degree of \( Q|_\rho \):}  If \( x_i \) is set to 0 in \( \rho \), then by induction
        \[
        \deg(Q|_\rho) = \deg\big((Q|_{x_i = 0})|_{\rho_0}\big) = \srho{\rho_0} = \srho{\rho}.
        \]
        If \( x_i \) is left free, then \( Q|_\rho = x_i \cdot (R_1|_{\rho_*}) + (R_0|_{\rho_*}) \), where \( Q = x_i R_1 + R_0 \). Hence:
        \[
        \deg(Q|_\rho) \overset{(1)}= \max\big(1 + \deg(R_1|_{\rho_*}), \deg(R_0|_{\rho_*})\big) \geq 1 + \deg(R_1|_{\rho_*}) \overset{(2)}{=} 1 + \srho{\rho_*} = \srho{\rho},\]

        where (1) uses the fact that \( R_1|_{\rho_*} \) is non-zero, and (2) applies the induction hypothesis to \( R_1 \). Since $\deg(Q|_\rho)$ cannot exceed $\srho{\rho}$, the inequality above must in fact be an equality. 
        Moreover, since both \( (Q|_{x_i = 0})|_{\rho_0} \) and \( R_1|_{\rho_*} \) are non-zero by the induction hypothesis, it follows that \( Q|_\rho \) is also non-zero.
        
        \item \textbf{(b) Degree of any monomial under restriction:}  
        Let \( M \) be any monomial.
        \begin{itemize}
        \item If \( x_i \notin \vars(M) \),  
        then: 
        \[
        \Pr_{\rho}(\deg(M|_{\rho}) \geq t) = \frac{1}{2} \Pr_{\rho_0}(\deg(M|_{\rho_0}) \geq t) + \frac{1}{2} \Pr_{\rho_*}(\deg(M|_{\rho_*}) \geq t) \leq \frac{2^{-t}}{2} + \frac{2^{-t}}{2} = 2^{-t},
        \]
        by the inductive hypothesis applied to both \( \rho_0 \) and \( \rho_* \).
        
        \item If \( x_i \in \vars(M) \), then  
        with probability \(1/2\), \( x_i \gets 0 \), so \( M|_{\rho} = 0 \). With the remaining probability \(1/2\), \( x_i \) remains free, and \( M|_\rho = x_i \cdot M'|_{\rho_*} \), where \( M' = M / x_i \). Hence:
                    \[
                    \Pr_\rho(\deg(M|_\rho) \geq t) = \tfrac{1}{2} \Pr_{\rho_*}(\deg(M'|_{\rho_*}) \geq t - 1) \leq \tfrac{1}{2} \cdot 2^{-(t - 1)} = 2^{-t}.
                    \]
        again using the inductive hypothesis.
        \end{itemize}

\end{itemize}
\end{enumerate}
This completes the proof.
\end{proof}

Applying Algorithm~\ref{algo:restriction} to the polynomial \( Q \) representing a Boolean function \( f \), and combining \cref{claim:sizeofrho} with \cref{claim:rhoProp}, we conclude that a large \( \ell_1 \)-norm implies the existence of a max-degree distribution.

\begin{theorem}\label{thm:sparsity-harddist}
Let \( f : \{0,1\}^n \to \{0,1\} \) be a Boolean function with \( \dmnorm{f} \geq 10(4n)^{40} \). Then there exists an \( \ell \)-variable max-degree distribution for \( f \), where \( \ell = \Omega\left(\frac{\log \dmnorm{f}}{\log n}\right) \).
\end{theorem}

\begin{proof}
Let \( Q \) be the unique multilinear real polynomial that exactly computes \( f \), and let \( \supp = \{x_1, \dots, x_n\} \). Consider the distribution over restrictions \( \rho \sim \samplerho(Q, \supp) \) generated by Algorithm~\ref{algo:restriction}. By \cref{claim:sizeofrho} and \cref{claim:rhoProp}, this distribution satisfies all three conditions of an \( \ell \)-variable max-degree distribution for \( \ell = \Omega\left(\frac{\log \dmnorm{f}}{\log n}\right) \).
\end{proof}

\subsection{Putting Everything Together} \label{subsec:together}

We begin by showing that the existence of a max-degree distribution for a Boolean function \( f \) implies that any polynomial approximating \( f \) must have large sparsity and large \( \ell_1 \)-norm. Combined with \cref{thm:sparsity-harddist}, which guarantees such a distribution when \( \dmnorm{f} \) is large, this implies the following: on a logarithmic scale, the exact sparsity and exact \( \ell_1 \)-norm are at most quadratically larger than their approximate counterparts, up to a \( \log n \) factor—proving \cref{thm:spar-aspar} and \cref{thm:ell-1}.

\begin{claim}\label{claim:harddist-sparsity}
Let \( f : \{0,1\}^n \to \{0,1\} \), and suppose there exists an \( \ell \)-variable max-degree distribution \( \mathcal{D} \) for \( f \). Then,
\[
    \log \aspar(f) = \Omega(\sqrt{\ell}).
\]
\end{claim}
\begin{proof}
Let \( \mathcal{D} \) be an \( \ell \)-variable max-degree distribution for \( f \). Suppose, for the sake of contradiction, that the claim does not hold. Let \( k = \sqrt{\ell / c} \), where \( c > 0 \) is a constant to be chosen later. Assume there exists a real polynomial \( Q \) that \( 1/3 \)-approximates \( f \) and has sparsity
\[
\spar(Q) \leq \frac{1}{10} \cdot 2^{k}.
\]
We will argue that such a polynomial cannot exist, thereby proving the claim.

Sample a restriction \( \rho \sim \mathcal{D} \), and consider the restricted polynomial \( Q|_\rho \). By property (3) of \( \mathcal{D} \), the probability that any fixed monomial in \( Q \) has degree at least \( k \) under \( \rho \) is at most \( 2^{-k} \). Applying a union bound over all monomials in \( Q \), we have
\[
\Pr_\rho\left(\deg(Q|_\rho) \geq k \right) \leq \spar(Q) \cdot 2^{-k} \leq \frac{1}{10}.
\]

By property (1) of \( \mathcal{D} \), with probability at least \( 0.9 \), \( \rho \) leaves at least \( \ell \) variables free. Thus, with probability at least \( 0.8 \), both of the following hold:
\[
\srho{\rho} \geq \ell \quad \text{and} \quad \deg(Q|_\rho) < k.
\]

Fix such a restriction \( \rho \). Then \( Q|_\rho \) is a polynomial of degree less than \( k \) that \( 1/3 \)-approximates \( f|_\rho \). By property (2) of \( \mathcal{D} \), we have:
\[
\deg(f|_\rho) = \srho{\rho} \geq \ell = c \cdot k^2 > c \cdot \left( \deg(Q|_\rho) \right)^2 \geq c \cdot \adeg(f|_\rho)^2.
\]
This contradicts the known relationship between degree and approximate degree for Boolean functions, namely that for all \( g \),
$\deg(g) \leq c\cdot \adeg(g)^2$ for some universal constant \( c \) (see \cref{thm:degree-apxdegree}).

Hence, our assumption was false, and the claim follows.
\end{proof}

\begin{claim}\label{claim:harddist-approxl1}
Let \( f : \{0,1\}^n \to \{0,1\} \), and suppose there exists an \( \ell \)-variable max-degree distribution \( \mathcal{D} \) for \( f \). Then,
\[
    \log \dmnormapprox{f}{} = \Omega(\sqrt{\ell}).
\]
\end{claim}

\begin{proof}
Let \( \mathcal{D} \) be an \( \ell \)-variable max-degree distribution for \( f \). Suppose, for the sake of contradiction, that the claim does not hold. Let \( k = (1 / c_1) \cdot \sqrt{\ell/c} \) for appropriate positive constants \( c \) and \( c_1 \) to be determined later. Assume there exists a real polynomial \( Q = \sum_{S \subseteq [n]} q_S \prod_{i \in S} x_i \) that \( 1/3 \)-approximates \( f \) and has \( \ell_1 \)-norm
\[
\dmnorm{Q} \leq \frac{1}{100} \cdot 2^k.
\]
We will argue that such a polynomial cannot exist, thereby proving the claim.

Sample a restriction \( \rho \) from \( \mathcal{D} \), and consider the restricted polynomial \( Q|_\rho \). We analyze the expected \( \ell_1 \)-mass of high-degree monomials in \( Q|_\rho \). For any polynomial \( P = \sum_{S \subseteq [n]} a_S \prod_{i \in S} x_i \), define the degree-\( d \) tail of its \( \ell_1 \)-norm as
\[
\dmnorm{P}^{\geq d} := \sum_{\substack{S \subseteq [n] \\ |S| \geq d}} |a_S|.
\]
Using property (3) of the distribution \( \mathcal{D} \), we get:
\[
\mathbb{E}_{\rho}\left[\dmnorm{Q|_\rho}^{\geq k}\right]
\leq \sum_{\substack{S \subseteq [n] \\ |S| \geq k}} |q_S| \cdot \Pr_{\rho}\left(\deg\left(\prod_{i \in S} x_i|_\rho\right) \geq k\right)
\leq \dmnorm{Q} \cdot 2^{-k}
\leq \frac{1}{100}.
\]

By Markov's inequality, with probability at least \( 0.9 \), we have \( \dmnorm{Q|_\rho}^{\geq k} < 0.1 \). Combining this with property (1) of \( \mathcal{D} \), which ensures \( \srho{\rho} \geq \ell \) with probability at least \( 0.9 \), we conclude that with probability at least \( 0.8 \), a random restriction \( \rho \sim \mathcal{D} \) satisfies both:
\[
\srho{\rho} \geq \ell
\quad \text{and} \quad
\dmnorm{Q|_\rho}^{\geq k} < 0.1.
\]

Fix such a restriction \( \rho \). Let \( \bar{Q} \) be the polynomial obtained from \( Q|_\rho \) by discarding all monomials of degree at least \( k \). Since \( Q|_\rho \) \( 1/3 \)-approximates \( f|_\rho \) and the total weight of the discarded tail is at most \( 0.1 \), it follows that \( \bar{Q} \) \( 0.44 \)-approximates \( f|_\rho \), with \( \deg(\bar{Q}) < k \).

By standard error reduction (see \cref{thm:error-reduction}), we can boost the success probability of \( \bar{Q} \) to obtain a polynomial that \( 1/3 \)-approximates \( f|_\rho \) with degree at most \( c_1 k \). Thus, \(\adeg(f|_\rho) < c_1 k = \sqrt{\ell / c} \). On the other hand, by property (2) of \( \mathcal{D} \), we have \( \deg(f|_\rho) = \srho{\rho} \geq \ell \). But this contradicts the known relationship between degree and approximate degree for Boolean functions,  which asserts that for any Boolean function \( g \), \( \deg(g) \leq c \cdot \adeg(g)^2 \) for some universal constant \( c \) (see \cref{thm:degree-apxdegree}).

Hence, our assumption was false, and the claim follows.
\end{proof}

\begin{theorem}\label{thm:l1-aspar}
For every total Boolean function $f : \{0,1\}^n \to \{0,1\}$, we have
\[
    \log(\dmnorm{f}) = O\big(\log(\aspar(f))^2 \cdot \log n\big).
\]
\end{theorem}

\begin{proof}
Assume \( \dmnorm{f} \geq 10(4n)^{40} \), as the claim is trivial otherwise. By \cref{thm:sparsity-harddist}, there exists an \( \ell \)-variable max-degree distribution \( \mathcal{D} \) for \( f \), where \( \ell = \Omega\left(\frac{\log \dmnorm{f}}{\log n}\right) \). Applying \cref{claim:harddist-sparsity} to \( \mathcal{D} \), we obtain
\[
\log \aspar(f) = \Omega(\sqrt{\ell}) = \Omega\left(\sqrt{\frac{\log \dmnorm{f}}{\log n}}\right). \qedhere
\]
\end{proof}

Since \( \spar(f) \leq \dmnorm{f} \) for any Boolean function \( f \) (see \cref{rmk:spar-l1}), the above theorem implies that for all total Boolean functions \( f \),
\[
    \log(\spar(f)) = O\big(\log^2(\aspar(f)) \cdot \log n\big),
\]
thereby proving \cref{thm:spar-aspar}.

Combining \cref{thm:sparsity-harddist} with \cref{claim:harddist-approxl1}, we also obtain:

\thmellnorm*

\begin{proof}
Assume \( \dmnorm{f} \geq 10(4n)^{40} \), as the claim is trivial otherwise. By \cref{thm:sparsity-harddist}, there exists an \( \ell \)-variable max-degree distribution \( \mathcal{D} \) for \( f \), where \( \ell = \Omega\left(\frac{\log \dmnorm{f}}{\log n}\right) \). Applying \cref{claim:harddist-approxl1} to \( \mathcal{D} \) yields the desired bound.
\end{proof}

\begin{remark}\label{rmk:grolmusz}
It is known that \( \log \aspar(f) = O(\log \dmnormapprox{f}{} + \log n) \), a result referred to as Grolmusz's theorem~\cite{grolmusz1997power,zhang2014efficient}, which has its roots in a paper by Bruck and Smolensky~\cite{bruck1992polynomial}. The works~\cite{grolmusz1997power,zhang2014efficient} show this bound for the Fourier basis, but the underlying proof technique is more general and, in fact, provides a method for converting a weighting measure to a counting measure. In particular, it applies to the De Morgan basis as well. Therefore, by combining this relationship with \cref{thm:l1-aspar}, we could have directly obtained a bound relating the exact and approximate \( \ell_1 \)-norms, without relying on \cref{claim:harddist-approxl1}. However, this approach yields a slightly weaker result, incurring an additive \( (\log n)^{O(1)} \) loss. Specifically, we would get
\(
\log \dmnorm{f} = O\left((\log \dmnormapprox{f}{})^2 \cdot \log n + (\log n)^3\right).
\)
\end{remark}

Finally, combining \cref{thm:l1-aspar,thm:spar-aspar,thm:ell-1}, we conclude that, on a logarithmic scale, the exact sparsity, approximate sparsity, exact \( \ell_1 \)-norm, and approximate \( \ell_1 \)-norm of any Boolean function are all polynomially related, up to polylogarithmic factors in \( n \).

\begin{remark}
Although we state our result for Boolean functions, the proof relies on minimal properties specific to Boolean-valuedness. In fact, the argument extends to any class of functions over the Boolean domain that is closed under variable fixing and satisfies a universal relationship between degree and approximate degree. In such settings, this relationship can be lifted to one between the logarithm of sparsity and the logarithm of approximate sparsity. In particular, since Boolean functions are closed under variable fixing and their degree and approximate degree are polynomially related, it follows that their sparsity and approximate sparsity are polynomially related on a logarithmic scale.
\end{remark}

\subsection{Discussion on the Optimality of Our Results}\label{subsec:spar-optimality}

\paragraph{Optimality of the bound in \cref{thm:spar-aspar}.}
Our result in \cref{thm:spar-aspar} is optimal up to polynomial factors in \( \log n \), as witnessed by the \( \OR_n \) function. The exact sparsity of \( \OR_n \) is \( 2^n - 1 \), while its approximate degree is \( \Theta(\sqrt{n}) \)~\cite{nisan1994degree}, implying that its approximate sparsity is at most \( n^{O(\sqrt{n})} \). Thus, 
$\log \spar(\OR_n) = \Theta(n)$ and  $\log \aspar(\OR_n) = O(\sqrt{n} \log n)$,  showing that the upper bound in \cref{thm:spar-aspar} is essentially tight up to polynomial factors in \( \log n \).

The dependence on \( n \) is also unavoidable. Consider the function \( \Thr^n_{n-1} : \{0,1\}^n \to \{0,1\} \), defined as
\[
\Thr^n_{n-1}(x) = 1 \quad \text{iff} \quad |x| \geq n-1,
\]
namely, the function evaluates to 1 if the input has at most one zero. Its exact sparsity is \( n + 1 \), via
\[
\Thr^n_{n-1}(x) = \sum_{\substack{S \subseteq [n] \\ |S| = n-1}} \prod_{i \in S} x_i - (n-1) \prod_{i \in [n]} x_i,
\]
while we show its approximate sparsity is only \( O(\log n) \), implying that an additive \( O(\log n) \) or multiplicative \( O\left(\frac{\log n}{\log \log n}\right) \) factor is necessary in \cref{thm:spar-aspar}.

\begin{claim}\label{claim:aspar-thr}
\( \aspar(\Thr^n_{n-1}) = O(\log n).\)
\end{claim}

\begin{proof}
We begin by introducing a combinatorial structure that underlies our construction.

\paragraph{Separating collections.}
Let \( \{i, j\} \in \binom{[n]}{2} \) be an unordered pair of distinct indices. A set \( S \subseteq [n] \) is said to \emph{separate} \( \{i,j\} \) if exactly one of \( i \) or \( j \) belongs to \( S \). A pair \( (S_1, S_2) \in 2^{[n]} \times 2^{[n]} \) is said to separate \( \{i,j\} \) if at least one of \( S_1 \) or \( S_2 \) separates it.

We say that a collection \( F \subseteq 2^{[n]} \times 2^{[n]} \) is \emph{\( \delta \)-separating} if, for every pair \( \{i,j\} \in \binom{[n]}{2} \), at least a \( \delta \)-fraction of the elements in \( F \) separate it. Formally,
\[
\forall \{i,j\} \in \binom{[n]}{2}, \quad \left| \left\{ (S_1, S_2) \in F : (S_1, S_2) \text{ separates } \{i,j\} \right\} \right| \geq \delta |F|.
\]

We will show that there exists a \( 2/3 \)-separating collection \( F \) of size \( O(\log n) \). Assuming such a collection exists, we describe a low-sparsity approximator for \( \Thr^n_{n-1} \).

\paragraph{Approximator construction.}
Let \( F \subseteq 2^{[n]} \times 2^{[n]} \) be a \( 2/3 \)-separating collection of size \( O(\log n) \). For each pair \( (S_1, S_2) \in F \), define
\[
f_{(S_1, S_2)}(x) = \left(1 - (1 - \prod_{i \in S_1} x_i)(1 - \prod_{i \notin S_1} x_i)\right) \cdot \left(1 - (1 - \prod_{i \in S_2} x_i)(1 - \prod_{i \notin S_2} x_i)\right).
\]
Each function \( f_{(S_1, S_2)} \) evaluates to 1 if the input \( x \in \{0,1\}^n \) contains at most one zero, and evaluates to 0 if, for some \( S \in \{S_1, S_2\} \), the input \( x \) contains zeros in both \( S \) and its complement. Define

\[
g(x) := \frac{1}{|F|} \sum_{(S_1, S_2) \in F} f_{(S_1, S_2)}(x).
\]
We claim that \( g \) is a \( 1/3 \)-approximator for \( \Thr^n_{n-1} \).

\begin{itemize}
\item If \( x \) is a 1-input, i.e., \( x \) has at most one zero, then for every \( S \subseteq [n] \), either \( \prod_{i \in S} x_i = 1 \) or \( \prod_{i \notin S} x_i = 1 \). Thus, each term \( f_{(S_1, S_2)}(x) = 1 \), so \( g(x) = 1 \).

\item If \( x \) is a 0-input, i.e., it contains at least two zeros, let \( i, j \in [n] \) be distinct positions where \( x_i = x_j = 0 \). For any \( (S_1, S_2) \in F \) that separates \( \{i, j\} \), one of \( S_1 \) or \( S_2 \) contains exactly one of \( i \), \( j \), so one of the products in the corresponding \( f_{(S_1, S_2)}(x) \) vanishes, and hence \( f_{(S_1, S_2)}(x) = 0 \). Since \( F \) is \( 2/3 \)-separating, at least \( 2/3 \) of the terms in the sum are 0, so \( g(x) \leq 1/3 \).
\end{itemize}

Hence, \( g \) is a \( 1/3 \)-approximator for \( \Thr^n_{n-1} \). Each function \( f_{(S_1, S_2)} \) has constant sparsity, and there are \( O(\log n) \) such terms in the sum, so the total sparsity of \( g \) is \( O(\log n) \).

\paragraph{Existence of separating collections.}
It remains to show that a \( 2/3 \)-separating collection of size \( O(\log n) \) exists. We do this via the probabilistic method.

Let \( t = 216 \ln(n^2) = O(\log n) \), and sample \( F = \{(S_1^{(k)}, S_2^{(k)})\}_{k=1}^t \), where each set \( S_u^{(k)} \subseteq [n] \) (for \( u \in \{1,2\} \)) is formed by including each element independently with probability \( 1/2 \). We show that with positive probability, \( F \) is \( 2/3 \)-separating. 

Fix a pair \( \{i,j\} \in \binom{[n]}{2} \), and let \( X_i \) be the indicator that \( (S_1^{(i)}, S_2^{(i)}) \) separates \( \{i,j\} \).  Each \( X_i \) has expectation \( \mathbb{E}[X_i] = 3/4 \), so the sum \( \sum_{i=1}^t X_i \) has expectation \( \frac{3t}{4} \). By a Chernoff bound,
\[
\Pr\left[\sum_{i=1}^t X_i \leq \frac{2t}{3}\right] \leq e^{-t/216} \leq \frac{1}{n^2}.
\]
Taking a union bound over all \( \binom{n}{2} < n^2 \) pairs, the probability that \( F \) fails to be \( 2/3 \)-separating for some pair is less than \( 1/2 \). Hence, with positive probability, a \( 2/3 \)-separating set \( F \) of size \( t = O(\log n) \) exists.
\end{proof}

\paragraph{Optimality of the bound in \cref{thm:ell-1}.}
The bound in \cref{thm:ell-1} is also tight up to polynomial factors in \( \log n \), again witnessed by \( \OR_n \). Its exact \( \ell_1 \)-norm is \( 2^n - 1 \), while its approximate \( \ell_1 \)-norm is at most \( n^{O(\sqrt{n})} \), as shown via a standard Chebyshev polynomial approximator.

\begin{observation}[\( \dmnormapprox{\OR_n}{} \leq n^{O(\sqrt{n})} \)]
Let \( T_d \) denote the degree-\( d \) Chebyshev polynomial defined recursively by \( T_0(z) = 1 \), \( T_1(z) = z \), and \( T_d(z) = 2zT_{d-1}(z) - T_{d-2}(z) \) for \( d \geq 2 \). For \( d = 2\sqrt{n} \), define
\[
p(z) = 1 - \frac{T_d\left(\frac{n - z}{n - 1}\right)}{T_d\left(\frac{n}{n - 1}\right)}, \quad \text{and} \quad q(x_1, \ldots, x_n) = p\left(\sum_{i=1}^n x_i\right).
\]
Then \( q \) \( 1/3 \)-approximates \( \OR_n \) (see \cite[Example 2]{nisan1994degree}). Furthermore, using the recursive definition, it is easy to verify that the coefficients of \( T_d \) are bounded in absolute value by \( 3^d \). Therefore, for \( d = 2\sqrt{n} \), the \( \ell_1 \)-norm of \( q \) is at most \( n^{O(\sqrt{n})} \).
\end{observation}

As with sparsity, the dependence on \( n \) in \cref{thm:ell-1} cannot be avoided as well. For \( \Thr^n_{n-1} \), the exact \( \ell_1 \)-norm is \( 2n - 1 \), while the approximator from Claim~\ref{claim:aspar-thr} has constant \( \ell_1 \)-norm. Thus, a \( \log n \) factor—either additive or multiplicative—is necessary.

\subsection{Implications for the $\AND$ Query Model} \label{subsec:spar-AND}
The measure \( \log \spar(f) \) naturally connects to the \(\AND\)-query model—a variant of the standard decision tree model where each query computes the $\AND$ of an arbitrary subset of input bits. Just as polynomial degree characterizes ordinary deterministic query complexity up to polynomial loss, Knop et al.~\cite{knop2021log} showed that \( \log \spar(f) \) characterizes deterministic query complexity in the \(\AND\)-query model, up to polynomial loss and polylogarithmic factors in \( n \).

In the randomized setting, it is easy to see that \( \log \aspar(f) \) lower bounds randomized \(\AND\)-query complexity. Let \( \RmAdt(f) \) denote the randomized \(\AND\)-query complexity of \( f \). The following is easy to verify:

\begin{claim}[{\cite[Claim 3.20]{knop2021guest}}]\label{claim:aspar-radt}
For every total Boolean function \( f : \{0,1\}^n \to \{0,1\} \), we have
\[
\log \aspar(f) = O(\RmAdt(f) + \log n).
\]
\end{claim}

However, it was unknown whether it also characterizes the randomized query complexity up to polynomial loss. Our results, combined with those of~\cite{knop2021log}, establish that this is indeed the case.

Knop et al.\ showed that for any Boolean function \( f \),
\[
\DmAdt(f) = O\left((\log \spar(f))^5 \cdot \log n\right),
\]
which, when combined with \cref{thm:spar-aspar}, implies
\[
\RmAdt(f) \leq \DmAdt(f) = O\left((\log \aspar(f))^{10} \cdot (\log n)^6\right).
\]

A tighter bound can be obtained using a structural result of Knop et al., which relates deterministic \(\AND\)-query complexity to sparsity and a combinatorial measure called \emph{monotone block sensitivity}:

\begin{definition}[Monotone Block Sensitivity]
The \emph{monotone block sensitivity} of a Boolean function \( f : \{0,1\}^n \to \{0,1\} \), denoted \( \mbs(f) \), is a variant of block sensitivity that only considers flipping 0’s to 1’s. A subset \( B \subseteq [n] \) is called a \emph{sensitive 0-block of \( f \) at input \( x \)} if \( x_i = 0 \) for all \( i \in B \), and \( f(x) \neq f(x \oplus 1_B) \), where \( x \oplus 1_B \) denotes the input obtained by flipping all bits in \( B \) from 0 to 1. For an input \( x \in \{0,1\}^n \), let \( \mbs(f,x) \) denote the maximum number of pairwise disjoint sensitive 0-blocks of \( f \) at \( x \). Then, \(\mbs(f) = \max_{x \in \{0,1\}^n} \mbs(f,x).
\)
\end{definition}

\begin{claim}[{\cite[Lemma 3.2, Claim 4.4, Lemma 4.6]{knop2021log}}]\label{claim:dadt-mbs-spar}
For any Boolean function \( f \),
\[
    \DmAdt(f) = O\left((\log \mbs(f))^2 \cdot \log \spar(f) \cdot \log n\right).
\]
\end{claim}

Intuitively, a large value of \( \mbs(f) \) indicates that a large-arity \( \pOR \) function can be embedded into \( f \) via suitable restrictions and identifications of variables.

To tighten our upper bound on \( \RmAdt(f) \), we now upper bound \( \mbs(f) \) in terms of \( \log \aspar(f) \). While Knop et al.\ showed \( \mbs(f) = O((\log \spar(f))^2) \), the same proof idea gives a similar bound in terms of approximate sparsity:

\begin{claim}\label{claim:mbs-aspar}
For any Boolean function \( f \),
\[
    \mbs(f) = O\left((\log \aspar(f))^2\right).
\]
\end{claim}
\begin{proof}
Assume \( \mbs(f) = k \geq 40 \); otherwise, the claim is trivial. Let this be witnessed by an input \( z \in \{0,1\}^n \) and disjoint $0$-blocks \( B_1, \dots, B_k \subseteq [n] \), such that \( f(z) \neq f(z \oplus 1_{B_i}) \) for all \( i \in [k] \).

Define \( g : \{0,1\}^k \to \{0,1\} \) by identifying variables within each \( B_i \), fixing all others according to \( z \), and letting \( g \) be the resulting function. Then \( g(0^k) = f(z) \) and \( g(x) \neq f(z) \) for all \( x \) with Hamming weight 1. Thus, \( g \) has sensitivity \( k \) at \( 0^k \). Since restrictions and identifications do not increase approximate sparsity, we have \( \aspar(g) \leq \aspar(f) \), so it suffices to show \( \aspar(g) \) is large.

Suppose, for contradiction, that \( g \) is \( 1/3 \)-approximated by a polynomial \( Q \) of sparsity
\[
\spar(Q) \leq \frac{1}{10} \cdot 2^{\ell},
\]
for \( \ell = c \cdot \sqrt{k/4} \), where \( c > 0 \) is a constant to be fixed later. We will argue that such a polynomial cannot exist, thereby proving the claim.

Define a distribution \( \mathcal{D} \) over restrictions \( \rho : \{x_1, \dots, x_k\} \to \{0,1,*\} \), where each variable is independently set to 0 with probability \( 1/2 \) and left free with probability \( 1/2 \). This distribution satisfies the following:

\begin{enumerate}
    \item By a standard Chernoff bound,
    \[
    \Pr_\rho\left[\srho{\rho} \leq k/4 \right] \leq e^{-k/16} \leq 0.1,
    \]
    where the last inequality uses \( k \geq 40 \). Thus, with probability at least \( 0.9 \), \( \srho{\rho} \geq k/4 \).

    \item For every \( \rho \) in the support of \( \mathcal{D} \), the restricted function \( g|_\rho \) has sensitivity \( \srho{\rho} \) at the all-zero input. Hence, by Theorem~\ref{thm:sens-adegree}, \( \adeg(g|_\rho) \geq c \cdot \sqrt{\srho{\rho}} \).
    \item For any monomial \( M \) over \( \{x_1, \dots, x_k\} \), we have \( \Pr_\rho[\deg(M|_\rho) > 0] = 2^{-\deg(M)}. \)
\end{enumerate}

Now consider the restricted polynomial \( Q|_\rho \). By property (3) of \( \mathcal{D} \), the probability that any fixed monomial in \( Q \) of degree at least \( \ell \) survives is at most \( 2^{-\ell} \), so by a union bound:
\[
\Pr_\rho[\deg(Q|_\rho) \geq \ell] \leq \spar(Q) \cdot 2^{-\ell} \leq \frac{1}{10}.
\]
By property (1) of \( \mathcal{D} \), with probability at least \( 0.9 \), \( \rho \) leaves at least \( k/4 \) variables free. Thus, with probability at least \( 0.8 \), both of the following hold:
\[
\srho{\rho} \geq k/4 \quad \text{and} \quad \deg(Q|_\rho) < \ell.
\]

Fix such a restriction \( \rho \). Then \( Q|_\rho \) is a polynomial of degree less than \( \ell \) that \( 1/3 \)-approximates \( f|_\rho \). Hence, \(\adeg(f|_{\rho}) < \ell = c \cdot \sqrt{k / 4} \leq c \cdot \sqrt{\srho{\rho}}\), contradicting property (2) of \( \mathcal{D} \). Hence, our assumption was false, and the claim follows.
\end{proof}

\sparsityADT*
\begin{proof}
The bound in (1) follows from \cref{claim:aspar-radt}. For (2), combining \cref{claim:dadt-mbs-spar}, \cref{claim:mbs-aspar}, and \cref{thm:spar-aspar}, we get:
\[
\DmAdt(f) = O\left((\log \mbs(f))^2 \cdot \log \spar(f) \cdot \log n\right) 
= O\left((\log \aspar(f))^6 \cdot \log n\right). \qedhere
\]
\end{proof}
This parallels the classical setting, where deterministic and randomized query complexity, degree, and approximate degree are all polynomially related. In the \(\AND\)-query model, \( \log \spar(f) \) plays the role of degree, while \( \log \aspar(f) \) plays the role of approximate degree. Combined with the results of~\cite{knop2021log}, our work shows that deterministic and randomized \(\AND\)-query complexities, log sparsity, and log approximate sparsity are all polynomially related—up to polylogarithmic loss factors.

\section{Exact vs Approximate Generalized Representations}

In this section, we present another application of our adaptive restriction technique. 
In the previous section, we showed that approximate polynomial representations of Boolean functions do not offer substantially more succinct representations than exact ones. Here, we study an analogous question in the setting of \emph{generalized polynomials}: 

\begin{quote}
\emph{Does allowing approximation lead to significantly sparser representations when using generalized polynomials?}
\end{quote}

We show that for \emph{monotone functions}, the answer is negative: their exact and approximate generalized sparsity and \( \ell_1 \)-norm are polynomially related on the logarithmic scale.

Unlike our result for standard polynomials, where we worked directly with polynomials exhibiting large exact sparsity or \(\ell_1\)-norm, it is unclear how to apply similar reasoning to generalized polynomials. A key obstacle is the non-uniqueness of generalized representations. For example, \( \OR_n(x_1, \ldots, x_n) \) can be expressed as \( \sum_{\emptyset \neq S \subseteq [n]} (-1)^{|S|} \prod_{i \in S} x_i \), which contains \( 2^n - 1 \) generalized monomials, or equivalently as \( 1 - \prod_{i = 1}^n (1 - x_i) \), which involves only 2 generalized monomials. This non-uniqueness makes it difficult to reason directly via generalized complexity measures. Therefore, our approach here proceeds differently: we identify a combinatorial structure that bridges the gap between exact and approximate generalized measures. Our argument proceeds in two steps:

\begin{itemize}
    \item First, we show that if a monotone function \( f \) has large exact generalized sparsity or \(\ell_1\)-norm, then it must have either a large set of maxterms or a large set of minterms.
    \item Second, we show that the existence of such a large set—either maxterms or minterms—implies large approximate generalized sparsity and large approximate generalized \( \ell_1 \)-norm.
\end{itemize}

Together, these implications show that for monotone functions, the exact and approximate generalized sparsity and $\ell_1$-norm cannot be too far apart.

The first implication is straightforward and follows from existing results. The second implication is our main technical contribution and is where we apply our adaptive restriction technique. In fact, we prove a more general result: we introduce a combinatorial notion called \emph{separating set}—a structural property of a set of inputs with respect to a function \( f \)—and show that any sufficiently large separating set necessitates high approximate generalized sparsity and \(\ell_1\)-norm. The case of a large set of minterms (or maxterms) is captured as a special instance of this framework.

\paragraph{Organization of this section.}
In~\cref{subsec:preliminaries}, we present the necessary preliminaries: we define generalized polynomials and their associated complexity measures, and recall some basic properties of monotone functions. In~\cref{subsec:separating}, we define the notion of separating set. In~\cref{subsec:separating-to-approx}, we establish that the existence of a separating set implies large approximate generalized sparsity and \(\ell_1\)-norm. In~\cref{subsec:sparg-asparg-monotone}, we put things together to conclude that the exact and approximate generalized measures are polynomially related on the logarithmic scale for monotone functions. Finally, in~\cref{subsec:gspar-impl}, we discuss implications for the decision tree size in the ordinary query model.

\subsection{Preliminaries}\label{subsec:preliminaries}

\paragraph{Generalized Polynomials.}

In standard polynomial representations, even simple functions like \( \OR_n = 1 - \prod_{i=1}^n (1 - x_i) \) can have high sparsity: the standard expansion of \( \OR_n \) contains \( 2^n - 1 \) monomials. To address this and allow for more compact representations, we consider \emph{generalized polynomials}, which extend standard polynomials by introducing formal complements \( \bar{x}_i \) for each variable \( x_i \). For example, \( \OR_n \) can be written more succinctly as
\[
\OR_n(x_1, \dots, x_n) = 1 - \prod_{i=1}^n \bar{x}_i,
\]
where each \( \bar{x}_i \) acts as a stand-in for \( 1 - x_i \). This representation uses only two monomials, offering exponential savings in sparsity.

We now define generalized polynomials formally.

\begin{definition}[Generalized Polynomial]
A \emph{generalized polynomial} is a polynomial over the ring
\[
    \mathbb{R}[x_1, \dots, x_n, \bar{x}_1, \dots, \bar{x}_n]/I,
\]
where \( \bar{x}_i \) denotes the formal complement of \( x_i \), and \( I \) is the ideal generated by the relations:
\[
    x_i^2 - x_i = 0 \quad \text{and} \quad x_i + \bar{x}_i - 1 = 0 \quad \text{for all } i \in [n].
\]
\end{definition}

\begin{definition}[Generalized Representation of Boolean Functions]
A generalized polynomial \( Q \in \mathbb{R}[x_1, \dots, x_n, \bar{x}_1, \dots, \bar{x}_n]/I \) \emph{represents} a function \( f : \{0,1\}^n \to \mathbb{R} \) if \( Q(x, \bar{x}) = f(x) \) for all \( x \in \{0,1\}^n \), where \( \bar{x}_i = 1 - x_i \).
\end{definition}

\begin{definition}[Generalized Complexity Measures]
As in the standard case, one can define the degree, sparsity, and \( \ell_1 \)-norm of a generalized polynomial. For a function \( f : \{0,1\}^n \to \mathbb{R} \), we define \( \gdeg(f) \), \( \gspar(f) \), and \( \gdmnorm{f} \) as the minimum degree, sparsity, and \( \ell_1 \)-norm, respectively, over all generalized polynomials that represent \( f \) exactly.

Analogously, the approximate measures \( \gadeg(f) \), \( \agspar(f) \), and \( \gdmnormapprox{f} \) denote the minimum degree, sparsity, and \( \ell_1 \)-norm among all generalized polynomials that approximate \( f \) pointwise within error \( 1/3 \).
\end{definition}
\begin{remark}\label{remark:gpoly}
Using generalized polynomials offers no advantage in terms of degree. Indeed, each dual variable \( \bar{x}_i \) can be replaced by \( 1 - x_i \), yielding a standard polynomial of the same degree. Therefore, \( \deg(f) = \gdeg(f) \) and \( \adeg(f) = \gadeg(f) \).
Since the degree measures coincide, we will simply write \( \deg(f) \) and \( \adeg(f) \) and avoid using the generalized notation \( \gdeg(f) \) and \( \gadeg(f) \).

On the other hand, generalized representations are not unique and can be exponentially more succinct. As discussed above, \( \OR_n \) has a generalized representation with just two monomials, while its standard representation requires \( 2^n - 1 \).
\end{remark}

\paragraph{Monotone Functions.}

For $x, y \in \{0,1\}^n$, we write $x \le y$ if $x_i \le y_i$ for all $i \in [n]$. A Boolean function $f : \{0,1\}^n \to \{0,1\}$ is \emph{monotone} if $x \le y$ implies $f(x) \le f(y)$.

\begin{definition}[Maxterms, Minterms, and Critical Inputs]
Let $f$ be a monotone Boolean function.
\begin{itemize}
    \item A \emph{maxterm} of $f$ is a minimal set $S \subseteq [n]$ such that setting all variables in $S$ to $0$ forces $f$ to output $0$. The associated \emph{critical 0-input} is the input $x_S \in \{0,1\}^n$ with $x_i = 0$ for $i \in S$ and $x_i = 1$ otherwise. Let $M_0(f) = \{x_S \mid S \text{ is a maxterm of } f\}$ denote the set of all such critical 0-inputs.
    
    \item A \emph{minterm} of $f$ is a minimal set $S \subseteq [n]$ such that setting all variables in $S$ to $1$ forces $f$ to output $1$. The associated \emph{critical 1-input} is the input $x_S$ with $x_i = 1$ for $i \in S$ and $x_i = 0$ otherwise. Let $M_1(f) = \{x_S \mid S \text{ is a minterm of } f\}$ denote the set of all such critical 1-inputs.
\end{itemize}
\end{definition}

\begin{observation}
For any monotone function $f$, every $x \in M_1(f)$ is sensitive on all $1$-bits: flipping any $1$ to $0$ changes $f(x)$ from $1$ to $0$. Similarly, every $x \in M_0(f)$ is sensitive on all $0$-bits.
\end{observation}

The following result of Ehrenfeucht and Haussler~\cite[Lemmas 1 and 6]{EH89} implies that if a monotone Boolean function has large exact generalized sparsity or $\ell_1$-norm, then the number of minterms (i.e., critical 1-inputs) or maxterms (i.e., critical 0-inputs) must be large.

\begin{theorem}[Ehrenfeucht and Haussler~\cite{EH89}]\label{thm:dsize_upper_bound_by_cover_number}
For every monotone Boolean function \( f : \{0,1\}^n \to \{0,1\} \), the following hold:
\begin{enumerate}
    \item [(a)] \( \log \gspar(f) = O\left(\log^2 M(f) \cdot \log n\right) \),
    \item [(b)] \( \log \gdmnorm{f} = O\left(\log^2 M(f) \cdot \log n\right) \),
\end{enumerate}
where \( M(f) = |M_0(f)| + |M_1(f)| .\) 
\end{theorem}

\begin{remark}\label{remark:gspar-dsize}
Ehrenfeucht and Haussler~\cite{EH89} proved a more general result: for any Boolean function $f$, they showed
\[
\log \DSize(f) = O(\log^2 \mathrm{Cover}(f) \cdot \log n),
\]
where $\DSize(f)$ is the size (number of leaves) of the smallest ordinary decision tree computing $f$, and $\mathrm{Cover}(f)$ is the minimum number of monochromatic subcubes under $f$ that cover $\{0,1\}^n$.

For monotone functions, the minimal subcube cover corresponds exactly to subcubes defined by maxterms and minterms, so $\mathrm{Cover}(f) = M(f)$. Hence, the above bound specializes to:
\[
\log \DSize(f) = O(\log^2 M(f) \cdot \log n).
\]

Moreover, any decision tree of size \( s \) and depth \( d \) computing a function \( f \) can be converted into a generalized polynomial of degree at most \( d \), and sparsity and \( \ell_1 \)-norm at most \( s \), that also computes \( f \). Hence,
\[
\gspar(f) \le \DSize(f) \quad \text{and} \quad \gdmnorm{f} \le \DSize(f).
\]
and the bound in \cref{thm:dsize_upper_bound_by_cover_number} follows.
\end{remark}

\subsection{Separating Sets of Inputs}\label{subsec:separating}

For any \( B \subseteq [n] \), let \( 1_B \in \{0,1\}^n \) denote the string with 1s in coordinates indexed by \( B \) and 0s elsewhere. For \( x \in \{0,1\}^n \) and \( i \in [n] \), we say that \( i \) is a \emph{sensitive coordinate} of \( f \) at \( x \) if \( f(x \oplus 1_{\{i\}}) \ne f(x) \). Let \( S(f,x) \subseteq [n] \) denote the set of all such coordinates. The \emph{sensitivity} of \( f \) is defined as \( \s(f) = \max_{x \in \{0,1\}^n} |S(f,x)| \).

We now define the notion of a \emph{separating} set of inputs for a Boolean function \( f \). Informally, its a collections of inputs that differ on their sensitive coordinates. Formally,

\begin{definition}[Separating Set of Inputs]
Let \( f : \{0,1\}^n \to \{0,1\} \) be a Boolean function. A set \( F \subseteq \{0,1\}^n \) is said to be \emph{separating} (with respect to \( f \)) if for every distinct pair \( x, y \in F \), the projections of \( x \) and \( y \) to the union of their sensitive coordinates differ; that is, letting \( B = S(f,x) \cup S(f,y) \), we require \( x|_B \ne y|_B \). We refer to this condition as the \emph{separation property}.
\end{definition}

We record some basic properties of separating sets that will be used later. First, the separation property is preserved under restrictions.

\begin{claim}[Closure under restriction]\label{claim:cons-closure}
Let \( F \subseteq \{0,1\}^n \) be a separating set with respect to \( f \), and let \( \rho \) be any restriction. Then the restricted set \( F|_\rho \) is separating with respect to the restricted function \( f|_\rho \).
\end{claim}

\begin{proof}
Suppose not. Then there exist \( x, y \in F|_\rho \) such that
\[
x|_B = y|_B \quad \text{for } B := S(f|_\rho, x) \cup S(f|_\rho, y).
\]
Let \( x', y' \in F \) be extensions of \( x, y \) consistent with \( \rho \). Then 
\[
x'|_{B'} = y'|_{B'} \quad \text{for } B' := S(f, x') \cup S(f, y') \subseteq B \cup \mathsf{SetVars}(\rho),
\]
contradicting the the separation property of \( F \).
\end{proof}

The separation property also ensures that a nontrivial fraction of inputs ($\approx 1/n$) in any separating set share a common sensitive coordinate:

\begin{claim}\label{claim:sensitive-density}
Let \( F \subseteq \{0,1\}^n \) be a separating set for \( f \) with \( |F| \ge 2 \). Then there exists an index \( i \in [n] \) and a subset \( F' \subseteq F \) such that \( i \in S(f, x) \) for every \( x \in F' \), and \( |F'| \ge |F| / (2n) \).
\end{claim}

\begin{proof}
If \( S(f,x) \neq \emptyset \) for all \( x \in F \), then by averaging, some \( i \in [n] \) appears in \( S(f,x) \) for at least \( |F|/n \) inputs; let \( F' := \{ x \in F : i \in S(f,x) \} \).

Otherwise, let \( x \in F \) have \( S(f,x) = \emptyset \). Separation property then forces every \( y \in F \setminus \{x\} \) to satisfy \( S(f,y) \neq \emptyset \). Averaging over \( F \setminus \{x\} \), some \( i \in [n] \) appears in \( S(f,y) \) for at least \( (|F|-1)/n \ge |F|/(2n) \) inputs. Let \( F' := \{ y \in F \setminus \{x\} : i \in S(f,y) \} \).
\end{proof}

We next observe that set of critical 1-inputs \( M_1(f) \) and the set of critical 0-inputs \( M_0(f) \) arising from monotone functions are separating:

\begin{claim}
Let \( f : \{0,1\}^n \to \{0,1\} \) be monotone. Then both the sets of critical 1-inputs \( M_1(f) \) and critical 0-inputs \( M_0(f) \) are separating.
\end{claim}

\begin{proof}
Consider distinct inputs \( x, y \in M_1(f) \). By the definition of critical 1-inputs, we have \( S(f,x) = \{ i \in [n] : x_i = 1 \} \), and \( S(f,y) = \{ i \in [n] : y_i = 1 \} \). Since \( x \ne y \), their sets of 1s differ, so there exists \( i \in S(f,x) \cup S(f,y) \) such that \( x_i \ne y_i \). Hence, \( x|_B \ne y|_B \) for \( B = S(f,x) \cup S(f,y) \), as required. A similar argument shows that \( M_0(f) \) is also separating.
\end{proof}

Combining this with \cref{thm:dsize_upper_bound_by_cover_number}, we conclude that any monotone function with large exact generalized sparsity or exact \( \ell_1 \)-norm must have a large separating set. Therefore, to relate the exact and approximate generalized sparsity and \( \ell_1 \)-norm of monotone functions, it suffices to show that large separating sets lead to large approximate sparsity and \( \ell_1 \)-norm. We establish this in the next section.

\subsection{Large Separating Set of Inputs Implies Large Approximate Generalized Sparsity and \(\ell_1\)-norm}\label{subsec:separating-to-approx}

In this subsection, we show that the existence of a large separating set for a function \( f \) implies that \( f \) must have large approximate generalized sparsity and \( \ell_1 \)-norm. We begin with an outline of the proof for the sparsity case; a similar argument also yields a lower bound on the approximate generalized \( \ell_1 \)-norm.

\paragraph{Proof Outline.}
The proof closely follows the template used earlier to relate sparsity and approximate sparsity for ordinary polynomials, with a key variation. This may be viewed as a second application of that general proof strategy.

Given a large separating set \( F \) with respect to a function \( f \), we sample a carefully designed random restriction \( \rho \), as described in Algorithm~\ref{algo2:restriction}, which satisfies the following key properties:
\begin{enumerate}
\item With high probability, \( \rho \) leaves a significant number of variables free—specifically, \( \ell = \Omega(\log |F| / \log n) \).
\item The restricted function \( f|_{\rho} \) has full sensitivity; that is, \( \s(f|_{\rho}) = \ell \).
\item For any generalized monomial, the probability that its degree under the restriction \( \rho \) exceeds \( t \) is at most \( 2^{-t} \).
\end{enumerate}

These properties lead us to conclude that any generalized polynomial approximating \( f \) must have large sparsity. Suppose, for contradiction, that there exists a polynomial \( Q \) that \( 1/3 \)-approximates \( f \) and has generalized sparsity at most \( 2^{c\sqrt{\ell}} \), for some constant \( c \). Then, using property (3) and a standard probabilistic argument, we can find a restriction \( \rho \) such that all generalized monomials of degree greater than \( c\sqrt{\ell} \) are eliminated. Consequently, the restricted polynomial \( Q|_{\rho} \) has degree at most \( c\sqrt{\ell} \) and approximate \( f|_{\rho} \). However, property (2) tells us that \( f|_{\rho} \) has sensitivity \( \ell \), and by the known relationship between sensitivity and approximate degree \cite{nisan1994degree}, this implies that any approximating polynomial for \( f|_{\rho} \) must have degree at least \( \Omega(\sqrt{\ell}) \). This contradicts the assumption that \( Q|_{\rho} \) has degree \( c\sqrt{\ell} \). Thus, any approximating generalized polynomial for \( f \) must have sparsity at least \( 2^{\Omega(\sqrt{\ell})} \), where \( \ell = \Omega(\log |F| / \log n) \).

Let us compare this restriction-based argument with the one from the previous section (Algorithm~\ref{algo:restriction}). Property (1) remains the same in spirit—the number of free variables is again \( \ell = \log |F| / \log n \), as opposed to \( \log \spar(f) / \log n \) in the earlier case. Property (2) differs: previously, the restricted function had full degree, while here it has full sensitivity. However, since we used large degree previously only to infer large approximate degree, we can make a similar inference here using the known connection between sensitivity and approximate degree. Hence, the role of property (2) in the argument remains essentially unchanged.

Property (3), on the other hand, is stronger than before. In the earlier argument, we could only reduce the degree of ordinary monomials, whereas here we are able to reduce the degree of even generalized monomials. This strengthens the conclusion of the argument.

As is evident, the overall structure of the argument closely mirrors the earlier one, with suitable adjustments to handle generalized monomials.

The following definition abstracts the essential properties of the random restriction process.

\begin{definition}[\( \ell \)-Variable Max-Sensitivity Distribution]\label{defn:harddist-sens}
Let \( f : \{0,1\}^n \to \{0,1\} \), and let \( \mathcal{D} \) be a distribution over restrictions \( \rho : \{x_1, \dots, x_n\} \to \{0,1,*\} \). We say that \( \mathcal{D} \) is an \emph{\( \ell \)-variable max-sensitivity distribution} for \( f \) if:
\begin{enumerate}
    \item With probability at least \( 0.9 \), \( \rho \) leaves at least \( \ell \) variables free.
    \item For every \( \rho \) in the support of \( \mathcal{D} \), we have \( \s(f|_\rho) = \srho{\rho} \).
    \item For any generalized monomial \( M \) and any \( t \in \mathbb{N} \), \( \Pr_{\rho \sim \mathcal{D}}[ \deg(M|_\rho) \geq t ] \leq 2^{-t} \).
\end{enumerate}
\end{definition}

We will show that a large separating set \( F \) implies the existence of such a distribution, which in turn implies that any generalized polynomial approximating \( f \) must have large sparsity and \( \ell_1 \)-norm.

\paragraph{Properties of the Restriction Process.}

\begin{algorithm}
\caption{\samplerhog}
\label{algo2:restriction}

\begin{algorithmic}[1]
\State \textbf{Input:} A set \( \supp \subseteq \{x_1, \dots, x_n\} \);  \( f:\{0,1\}^{\supp} \to \{0,1\} \); a nonempty set \( F \subseteq \{0,1\}^{\supp} \) separating w.r.t \( f \)
\State \textbf{Output:} A restriction \( \rho: \supp \to \{0,1,*\} \).

\If{\( |F| \leq 2 \)} \label{algo2:line:terminate} 
   \State Let \( w \) be an input in \( F \)
    \State For each \( x_i \in \supp \), set \( \rho(x_i) \gets w_i \) 
\Else
    \If{there exists \( x_i \in \supp \), \( u \in \{0,1\} \) such that \( |F|_{x_i = u}| \geq \left(1 - \frac{1}{n} \right) \cdot |F| \)} \label{algo2:line:cond2}
        \State \( \rho' \gets \samplerhog(\supp \setminus \{x_i\}, f|_{x_i = u}, F|_{x_i = u} ) \) \label{algo2:line:recursivecall1}
        \State Set \( \rho(x_i) \gets u \), and for all \( x_j \in \supp \setminus \{x_i\} \), set \( \rho(x_j) \gets \rho'(x_j) \)
    \Else \label{algo2:line:cond3}
        \State Choose \( x_i \in \supp \) and \( u \in \{0,1\} \) such that $\left| \left\{ w \in F \,\middle|\, i \in S(f, w) \text{ and } w_i = u \right\} \right| \geq \frac{|F|}{4n}$  \label{algo2:line:chosenVar}
        \State With probability \( 1/3 \): \label{algo2:line:rand0}
            \State \hspace{1.5em} \( \rho_0 \gets \samplerhog(\supp \setminus \{x_i\}, f|_{x_i = 0}, F|_{x_i = 0} ) \) \label{algo2:line:recursivecall2}
            \State \hspace{1.5em} Set \( \rho(x_i) \gets 0 \), and for all \( x_j \in \supp \setminus \{x_i\} \), set \( \rho(x_j) \gets \rho_0(x_j) \)
        \State With probability \( 1/3 \): \label{algo2:line:rand1}
            \State \hspace{1.5em} \( \rho_1 \gets \samplerhog(\supp \setminus \{x_i\}, f|_{x_i = 1}, F|_{x_i = 1} ) \) \label{algo2:line:recursivecall3}
            \State \hspace{1.5em} Set \( \rho(x_i) \gets 1 \), and for all \( x_j \in \supp \setminus \{x_i\} \), set \( \rho(x_j) \gets \rho_1(x_j) \)
        \State Otherwise: \label{algo2:line:rand-free}
            \State \hspace{1.5em} $F' \gets  \left\{ w \in F \,\middle|\, i \in S(f, w) \right\} $ 
            \State \hspace{1.5em} \( \rho_* \gets \samplerhog(\supp \setminus \{x_i\}, f|_{x_i = u}, F'|_{x_i = u} ) \) \label{algo2:line:recursivecall4}
            \State \hspace{1.5em} Set \( \rho(x_i) \gets * \), and for all \( x_j \in \supp \setminus \{x_i\} \), set \( \rho(x_j) \gets \rho_*(x_j) \)
    \EndIf
\EndIf
\State \Return \( \rho \)
\end{algorithmic}
\end{algorithm}

Algorithm~\ref{algo2:restriction} describes how to sample random restrictions for a given separating set \( F \) w.r.t \( f \). We will show that the resulting distribution over restrictions is \( \ell \)-variable max-sesitivity for \( f \), where \( \ell = \Omega\left(\frac{\log |F|}{\log n}\right) \).

We begin with some observations concerning the validity of \cref{algo2:restriction}.
The step at line~\ref{algo2:line:chosenVar}, which selects an $x_i\in \supp$, is justified by the following observation.

\begin{observation}\label{algo2:obs2}
By \cref{claim:sensitive-density}, at line~\ref{algo2:line:chosenVar}, Algorithm~\ref{algo2:restriction} is guaranteed to find \( x_i \in \supp \) such that
\[
\left| \left\{ w \in F \,\middle|\, i \in S(f, w) \right\} \right| \geq \frac{|F|}{2n}.
\]
Moreover, choosing \( u \in \{0,1\} \) to be the more frequent value of the \( i \)-th coordinate among these inputs ensures that
\[
\left| \left\{ w \in F \,\middle|\, i \in S(f, w) \text{ and } w_i = u \right\} \right| \geq \frac{|F|}{4n}.
\]
\end{observation}

We next address the correctness of the recursive structure. Since the separation property is preserved under restrictions, the recursive calls made on lines~\ref{algo2:line:recursivecall1}, \ref{algo2:line:recursivecall2}, \ref{algo2:line:recursivecall3}, and~\ref{algo2:line:recursivecall4} satisfy the preconditions of Algorithm~\ref{algo2:restriction}.

We now observe a key structural property of the restriction \( \rho \) returned by the algorithm, which shows property~(2) of the max-sensitivity distribution.

\begin{claim}\label{claim:rhoform}
Let \( \rho \) be the final restriction returned by Algorithm~\ref{algo2:restriction} on input \( (\supp, f, F) \). Then there exists an input \( w \in F \) such that:
\[
\forall x_i \in \mathsf{SetVars}(\rho),\ \rho(x_i) = w_i, \quad \text{and} \quad \forall x_i \in \mathsf{FreeVars}(\rho),\ i \in S(f, w).
\]
\end{claim}

\begin{proof}
We prove the claim by induction on the size of the \(\supp\) set, which reduces with each recursive call. 

\paragraph{Base case (\( |\supp| = 1 \)).}
In this case, the input set \( F \) must have size at most 2. The algorithm simply fixes the lone variable in \(\supp\) to either 0 or 1 based on some \( w \in F \), and returns the corresponding restriction. The claim then holds trivially for this choice of \( w \).

\paragraph{Inductive step (\( |\supp| > 1 \)).}
We consider the possible return paths in the algorithm depending on the branching conditions and the randomness involved (lines~\ref{algo2:line:terminate}, \ref{algo2:line:cond2}, \ref{algo2:line:rand0}, \ref{algo2:line:rand1}, and~\ref{algo2:line:rand-free}):

\begin{enumerate}
    \item \textbf{Case \(|F| \leq 2\):} This is similar to the base case. The algorithm returns a restriction that sets all variables in \(\supp\) to match some \( w \in F \), and the claim follows directly.

    \item \textbf{Case: branch taken via line~\ref{algo2:line:cond2}, \ref{algo2:line:rand0}, or \ref{algo2:line:rand1}.}
    In these branches, the algorithm chooses some \( x_i \in \supp \) and \( u \in \{0,1\} \), and returns a restriction of the form $\rho = \rho' \cup \{x_i \gets u\}$, where
    \[
    \rho' \gets \samplerhog(\supp \setminus \{x_i\}, f|_{x_i = u}, F|_{x_i = u}).
    \]
    By the inductive hypothesis applied to \(\rho'\), there exists \( w' \in F|_{x_i = u} \) such that
    \[
    \forall x_j \in \mathsf{SetVars}(\rho'),\ \rho'(x_j) = w'_j, \quad \text{and} \quad \forall x_j \in \mathsf{FreeVars}(\rho'),\ j \in S(f|_{x_i = u}, w').
    \]
    Let \( w \in F \) be an extension of \( w' \) with \( x_i = u \). Then \( \rho(x_j) = w_j \) for all set variables \( x_j \), and for all free variables \( x_j \), we have \( j \in S(f, w) \) since \( j \in S(f|_{x_i = u}, w') \). Thus, the claim holds for this \( w \).

    \item \textbf{Case: branch taken via line~\ref{algo2:line:rand-free}.}
    In this case, for \( x_i \in \supp \) and \( u \in \{0,1\} \) selected in line~\ref{algo:line:chosenVar}, the $\rho$ returned is $\rho_* \cup \{x_i \gets *\}$ where 
    \[
    \rho_* \gets \samplerhog(\supp \setminus \{x_i\}, f|_{x_i = u}, F'|_{x_i = u}).
    \]
    for \(F' \gets \left\{ w \in F \mid i \in S(f, w) \right\} \). By the inductive hypothesis applied to \( \rho_* \), there exists \( w' \in F'|_{x_i = u} \) such that
    \[
    \forall x_j \in \mathsf{SetVars}(\rho_*),\ \rho_*(x_j) = w'_j, \quad \text{and} \quad \forall x_j \in \mathsf{FreeVars}(\rho_*),\ j \in S(f|_{x_i = u}, w').
    \]
    Let \( w \in F' \) be an extension of \( w' \) with \( x_i = u \). Since \( i \in S(f, w) \), by definition of \(F'\), and since \( x_i \) is free in \( \rho \), the required conditions are satisfied by \( w \) for the final restriction \( \rho = \rho_* \cup \{x_i \gets *\} \).
\end{enumerate}

This completes the inductive proof.
\end{proof}

As a consequence, for any restriction \( \rho \) in the support of the distribution induced by Algorithm~\ref{algo2:restriction}, the restricted function \( f|_\rho \) has sensitivity exactly \( \srho{\rho} \), thereby satisfying property~(2) of the max-sensitivity distribution.

We next show that the restriction \( \rho \) produced by Algorithm~\ref{algo2:restriction} leaves a significant number of variables free, establishing property~(1) of the max-sensitivity distribution. The argument closely mirrors that of the earlier restriction algorithm.

Algorithm~\ref{algo2:restriction} proceeds recursively and follows one of the following three branches:
\begin{itemize}
    \item If \( |F| \leq 2 \), the algorithm terminates and returns \( \rho \) immediately (line~\ref{algo2:line:terminate}).
    \item If the condition on line~\ref{algo2:line:cond2} is satisfied, the algorithm makes a single recursive call (line~\ref{algo2:line:recursivecall1}).
    \item Otherwise (line~\ref{algo2:line:cond3}), the algorithm selects one of three recursive calls (lines~\ref{algo2:line:recursivecall2}, \ref{algo2:line:recursivecall3}, \ref{algo2:line:recursivecall4}) uniformly at random, each with probability \( 1/3 \).
\end{itemize}

The recursion continues as long as \( |F| > 2 \), and halts when \( |F| \leq 2 \), at which point the algorithm backtracks to construct the final restriction.

As in earlier analyses, we classify recursive calls as either \emph{active} (when the condition on line~\ref{algo2:line:cond3} holds) or \emph{passive} (when the condition on line~\ref{algo2:line:cond2} holds). We argue that a substantial fraction of the calls must be active, and in each such call, a variable is left free with probability \( 1/3 \). Moreover, these choices are made independently across calls. Therefore, by standard concentration bounds, the number of free variables in the final restriction is close to its expectation, which is 1/3 the number of active calls. This leads to the following claim:

\begin{claim}\label{claim:sizeofrho2}
Let \( (\supp, f, F) \) satisfy the input requirements of Algorithm~\ref{algo2:restriction}, and suppose that \( |F| \geq 20(4n)^{60} \). Then, with probability at least \( 0.9 \), the restriction \( \rho = \samplerhog(\supp, f, F) \) produced by Algorithm~\ref{algo2:restriction} leaves at least \( \Omega\left(\frac{\log |F|}{\log n}\right) \) variables free.
\end{claim}

\begin{proof}
We begin by showing that every execution of the algorithm encounters a substantial number of \emph{active} recursive calls. Suppose the algorithm makes \( t \) recursive calls in total, of which \( \ell \) are active. Since the size of \( \supp \) decreases by 1 with each recursive call, and the algorithm halts when \( |\supp| = 1 \) (which corresponds to \( |F| \leq 2 \)), we have \( t \leq n \).

We observe how the size of the separating set \( F \) evolves during recursion:
\begin{itemize}
    \item In a \emph{passive} call (i.e., when line~\ref{algo2:line:cond2} is satisfied), the size of \( F \) decreases by at most a factor of \( 1 - \frac{1}{n} \).
    
    \item In an \emph{active} call (i.e., when line~\ref{algo2:line:cond3} is satisfied), the reduction in the size of \( F \) depends on the specific recursive branch taken:
    \begin{itemize}
        \item For recursive calls on lines~\ref{algo2:line:recursivecall2} and \ref{algo2:line:recursivecall3}, the size of \( F \) decreases by at most a factor of \( 1/n \), owing to the balancedness condition enforced by line~\ref{algo2:line:cond3}.
        
        \item For the recursive call on line~\ref{algo2:line:recursivecall4}, the size of \( F \) decreases by at most a factor of \( 1/(4n) \), due to the choice of index \( i \) in line~\ref{algo2:line:chosenVar} (see \cref{algo2:obs2}).
    \end{itemize}
    Therefore, in any \emph{active} call, the size of \( F \) decreases by at most a factor of \( 1/(4n) \), regardless of which of the three recursive branches is chosen.
\end{itemize}

Since the algorithm halts and backtracks when the $|F|\leq 2$ , we obtain the following inequality:
\begin{align*}
2 &\geq (1 - 1/n)^{t - \ell} \cdot (1/4n)^{\ell} \cdot |F| \\
  &\geq (1 - 1/n)^n \cdot (1/4n)^{\ell} \cdot |F|
  && \text{(since \( t - \ell \leq n \))} \\
  &\geq \frac{1}{10} \cdot \left(\frac{1}{4n}\right)^{\ell} \cdot |F|
  && \text{(using \( (1 - 1/n)^n \geq 1/10 \) for \( n \geq 2 \))}.
\end{align*}

Taking logarithms and rearranging, we obtain:
\[
\ell \geq \frac{\log(|F|/20)}{\log(4n)}.
\]
Define \( \ell^* := \frac{\log(|F|/20)}{\log(4n)} \). Thus, every run of the algorithm contains at least \( \ell^* \) active recursive calls.

Let \( X_1, X_2, \dots, X_{\ell^*} \) be indicator random variables, where \( X_i = 1 \) if the variable chosen in the \( i \)-th active call is left free (which occurs with probability \( 1/3 \)), and \( 0 \) otherwise. These variables are independent by construction, and \( \mathbb{E}\left[\sum_{i=1}^{\ell^*} X_i\right] = \frac{\ell^*}{3}.\) By a standard Chernoff bound, we get:
\[
\Pr\left(\sum_{i=1}^{\ell^*} X_i \leq \frac{\ell^*}{6}\right) \leq e^{-\ell^*/24} \leq 0.1,
\]
where the last inequality follows from the assumption \( |F| \geq 20(4n)^{60} \).

Since the number of variables left free in the final restriction is at least \( \sum_{i=1}^{\ell^*} X_i \), we conclude that with probability at least \( 0.9 \), the algorithm leaves at least \( \ell^*/6 = \Omega\left(\frac{\log |F|}{\log n}\right) \) variables free.\qedhere

\end{proof}

Finally, we establish property~(3) of the max-sensitivity distribution, showing that for any generalized monomial \( M \), the degree of \( M|_\rho \) under the sampled restriction \( \rho \) exhibits exponential tail decay.

\begin{claim}\label{claim:rhoProp2}
Let \( \rho = \samplerhog(\supp, f, F) \) be the restriction returned by Algorithm~\ref{algo:restriction}. Then, for any generalized monomial \( M \) with \( \vars(M) \subseteq \{x_i \mid x_i \in \supp\} \), and any \( t \in \mathbb{N} \), we have:
\[
\Pr_\rho\left( \deg(M|_\rho) \geq t \right) \leq 2^{-t}.
\]
\end{claim}

\begin{proof}
We prove the claim by induction on \( |\supp| \).

\paragraph{Base case (\( |\supp| = 1 \)).}
Here, we must have \( |F| \leq 2 \), the restriction $\rho$ returned by the algorithm  sets all variables in \( \supp \) to either 0 or 1, based on some \(w\in F\). Thus, any monomial \( M \) becomes a constant under \( \rho \). The claim follows.

\paragraph{Inductive step (\( |\supp| > 1 \)).}
We consider the behavior of the algorithm based on the three possible branches (lines~\ref{algo2:line:terminate}, \ref{algo2:line:cond2}, and \ref{algo2:line:cond3}):

\begin{enumerate}
    \item \textbf{Case \(|F| \leq 2\):} Similar to the base case, \(\rho\) returned by the algorithm  sets all variables in \( \supp \) to either 0 or 1, based on some \(w\in F\), and hence \( \deg(M|_\rho) = 0 \).

    \item \textbf{Case where the ``if'' condition (line~\ref{algo2:line:cond2}) is satisfied:} 
    Suppose the condition is satisfied for some \( x_i \in \supp \) and some value \( u \in \{0,1\} \). The algorithm then returns the restriction \( \rho = \rho' \cup \{x_i \gets u\} \),  where \( \rho' \) is obtained from a recursive call with a strictly smaller support set. By the inductive hypothesis applied to \( \rho' \), for any generalized monomial \( M \), we have: 
    \[
    \Pr_\rho\left( \deg(M|_\rho) \geq t \right) \leq \Pr_{\rho'}\left( \deg(M|_{\rho'}) \geq t \right) \leq 2^{-t}.
    \]

    \item \textbf{Case: the \texttt{else} clause at line~\ref{algo:line:cond3} is executed.}  
Let \( x_i \in \supp \) be the variable selected in line~\ref{algo:line:chosenVar}, and let \( \rho_0, \rho_1, \rho_* \) be restrictions sampled from the recursive calls at lines~\ref{algo2:line:recursivecall2}, \ref{algo2:line:recursivecall3}, and \ref{algo2:line:recursivecall4}, respectively. For a generalized monomial \( M \), we consider two cases:

\begin{itemize}
  
  \item \textbf{\( M \) does not contain \( x_i \) or \( \bar{x}_i \):}  
  In this case, the monomial is unaffected by the assignment to \( x_i \). By applying induction hypothesis, we have:
  \[
  \Pr_{\rho}(\deg(M|_{\rho}) \geq t) = \tfrac{1}{3} \Pr_{\rho_0}(\deg(M|_{\rho_0}) \geq t) + \tfrac{1}{3} \Pr_{\rho_1}(\deg(M|_{\rho_1}) \geq t) + \tfrac{1}{3} \Pr_{\rho_*}(\deg(M|_{\rho_*}) \geq t) \leq 2^{-t} .
  \]
  
  \item \textbf{\( M \) contains \( x_i \) or \( \bar{x}_i \):}  
  Without loss of generality, suppose \( x_i \in M \) (the case \( \bar{x}_i \in M \) is analogous). Let \( M' = M / x_i \). Then:
        \[
        \Pr_\rho(\deg(M|_\rho) \geq t) = \frac{1}{3} \cdot 0 + \frac{1}{3} \Pr_{\rho_1}(\deg(M'|_{\rho_1}) \geq t) + \frac{1}{3} \Pr_{\rho_*}(\deg(M'|_{\rho_*}) \geq t-1),
        \]
        where the first term is zero because setting \( x_i \gets 0 \) kills the monomial. Using the inductive hypothesis:
        \[
        \Pr_\rho(\deg(M|_\rho) \geq t) \leq \frac{1}{3} \cdot 2^{-t} + \frac{1}{3} \cdot 2^{-(t-1)} = 2^{-t}.
        \]
\end{itemize}

\end{enumerate}
This completes the inductive proof.
\end{proof}

Applying Algorithm~\ref{algo2:restriction} to a separating set \( F \) for \( f \), and combining \cref{claim:rhoform}, \cref{claim:sizeofrho2}, and \cref{claim:rhoProp2}, we conclude that the existence of a large separating set implies the existence of a max-sensitivity distribution. This yields the following theorem.

\begin{theorem}\label{thm:separating-harddist}
Let \( f : \{0,1\}^n \to \{0,1\} \) be a Boolean function, and let \( F \subseteq \{0,1\}^n \) be a separating set for \( f \) with \( |F| \ge 20(4n)^{60} \). Then there exists an \( \ell \)-variable max-sensitivity distribution for \( f \), where \( \ell = \Omega\left(\frac{\log |F|}{\log n}\right) \).
\end{theorem}
\paragraph{Putting Everything Together.}
We first show that the existence of a max-sensitivity distribution forces any generalized polynomial approximating \( f \) to have large sparsity and \( \ell_1 \)-norm. Combined with \cref{thm:separating-harddist}, this implies that a large separating set yields large approximate generalized sparsity and large approximate generalized \( \ell_1 \)-norm.

We will use the following classical result of Nisan and Szegedy~\cite{nisan1994degree}, which relates the approximate degree of a Boolean function to its sensitivity.

\begin{theorem}[Nisan and Szegedy~\cite{nisan1994degree}]\label{thm:sens-adegree}
Let \( f : \{0,1\}^n \to \{0,1\} \) be a Boolean function. Then,
\[
\adeg(f) \geq \sqrt{\frac{\s(f)}{6}}.
\]
\end{theorem}

\begin{claim}\label{claim:harddist-gsparsity}
Let \( f : \{0,1\}^n \to \{0,1\} \), and suppose there exists an \( \ell \)-variable max-sensitivity distribution \( \mathcal{D} \) for \( f \). Then,
\[
    \log \agspar(f) = \Omega(\sqrt{\ell}).
\]
\end{claim}
\begin{proof}
Let \( \mathcal{D} \) be an \( \ell \)-variable max-sensitvity distribution for \( f \). Suppose, for contradiction, that the claimed bound does not hold. Let \( k = \sqrt{\ell / 6} \) and assume there exists a real generalized polynomial  \( Q \) that \( 1/3 \)-approximates \( f \) and has sparsity
\[
\spar(Q) \leq \frac{1}{10} \cdot 2^{k}.
\] 
We will argue that such a polynomial cannot exist, thereby proving the claim.

Sample a restriction \( \rho \sim \mathcal{D} \), and consider the restricted polynomial \( Q|_\rho \). By property (3) of \( \mathcal{D} \), the probability that any fixed generalized monomial in \( Q \) has degree at least \( k \) under \( \rho \) is at most \( 2^{-k} \). Applying a union bound over all monomials in \( Q \), we have
\[
\Pr_\rho\left(\deg(Q|_\rho) \geq k \right) \leq \spar(Q) \cdot 2^{-k} \leq \frac{1}{10}.
\]

By property (1) of \( \mathcal{D} \), with probability at least \( 0.9 \), \( \rho \) leaves at least \( \ell \) variables free. Thus, with probability at least \( 0.8 \), both of the following hold:
\[
\srho{\rho} \geq \ell \quad \text{and} \quad \deg(Q|_\rho) < k.
\]
Fix such a restriction \( \rho \). Then \( Q|_\rho \) is a polynomial of degree less than \( k \) that \( 1/3 \)-approximates \( f|_\rho \). By property (2) of \( \mathcal{D} \), we have:
\[
\s(f|_\rho) = \srho{\rho} \geq \ell = 6 \cdot k^2 > 6 \cdot \left( \deg(Q|_\rho) \right)^2 \geq 6 \cdot \adeg(f|_\rho)^2.
\]
This contradicts the relationship between approximate degree and sensitivity \cref{thm:sens-adegree}. Hence, our assumption was false, and the claim follows.
\end{proof}

Combining the above with Grolmusz’s theorem~\cite{grolmusz1997power,zhang2014efficient} (see \cref{rmk:grolmusz}), which gives
\[
\log \agspar(f) = O\left( \log \gdmnormapprox{f} + \log n \right),
\]
we also obtain a lower bound on the approximate generalized \( \ell_1 \)-norm from the existence of a max-sensitivity distribution.  While this approach incurs an extra additive \( \log n \) loss compared to the bound above, it can be avoided by directly arguing as in \cref{claim:harddist-approxl1}. Since the proof involves no new ideas, we omit the details for brevity and state the resulting optimal bound:

\begin{claim}\label{claim:harddist-glnorm}
Let \( f : \{0,1\}^n \to \{0,1\} \), and suppose there exists an \( \ell \)-variable max-sensitivity distribution \( \mathcal{D} \) for \( f \). Then,
\[
    \log \gdmnormapprox{f} = \Omega(\sqrt{\ell}).
\]
\end{claim}

Combining \cref{thm:separating-harddist} with \cref{claim:harddist-gsparsity} and \cref{claim:harddist-glnorm}, we obtain the following consequences of the existence of a large separating set:

\begin{theorem}\label{thm:cons-apxgspar}
Let \( f : \{0,1\}^n \to \{0,1\} \) be a Boolean function, and let \( F \subseteq \{0,1\}^n \) be a separating set with respect to \( f \). Then,
\[
\log \agspar(f) = \Omega\left( \left( \frac{\log |F|}{\log n} \right)^{1/2} \right).
\]
\end{theorem}


\begin{theorem}\label{thm:cons-apxgl1}
Let \( f : \{0,1\}^n \to \{0,1\} \) be a Boolean function, and let \( F \subseteq \{0,1\}^n \) be a separating set with respect to \( f \). Then,
\[
\log \gdmnormapprox{f} = \Omega\left( \left( \frac{\log |F|}{\log n} \right)^{1/2} \right).
\]
\end{theorem}

\subsection{Exact vs Approximate Generalized Measures for Monotone Functions}\label{subsec:sparg-asparg-monotone}

We now relate the generalized sparsity and generalized \( \ell_1 \)-norm of a monotone function to their approximate counterparts. This is done by combining \cref{thm:dsize_upper_bound_by_cover_number} with the results from the previous section.

\thmgspar*

\begin{proof}
From \cref{thm:dsize_upper_bound_by_cover_number}, we have:
\[
\log M(f) = \Omega\left(\sqrt{\frac{\log \gspar(f)}{\log n}}\right)
\quad \text{and} \quad
\log M(f) = \Omega\left(\sqrt{\frac{\log \gdmnorm{f}}{\log n}}\right).
\]
This implies that either the number of critical 1-inputs or the number of critical 0-inputs is large. Without loss of generality, assume \( |M_1(f)| \geq M(f)/2 \). 

Since both the sets of critical 1-inputs and critical 0-inputs are separating with respect to \( f \), we can apply \cref{thm:cons-apxgspar} and \cref{thm:cons-apxgl1} to the set \( M_1(f) \), yielding:
\[
\log \gspar(f) = O\left(\log^2 |M_1(f)| \cdot \log n\right)
= O\left((\log \agspar(f))^4 \cdot (\log n)^3\right),
\]
\[
\log \gdmnorm{f} = O\left(\log^2 |M_1(f)| \cdot \log n\right)
= O\left((\log \gdmnormapprox{f})^4 \cdot (\log n)^3\right),
\]
as claimed.
\end{proof}

\subsection{Implications for Decision Tree Size in the Ordinary Query Model} \label{subsec:gspar-impl}

The measures \( \log \gspar(f) \) and \( \log \gdmnorm{f} \) are related to the decision tree size \( \log \DSize(f) \) in the ordinary query model, as noted in \cref{remark:gspar-dsize}. Specifically, any decision tree of size \( s \) computing \( f \) can be converted into a generalized polynomial for \( f \) with sparsity and \( \ell_1 \)-norm at most \( s \).

For monotone functions, applying our result \cref{thm:mon-gen-sparsity} together with \cref{remark:gspar-dsize}, we obtain:
\corrgenspar*
\begin{proof}
For (1) of both (a) and (b), assume $\RSize(f) = s$, i.e., there exists a distribution $\mathcal{D}$ over decision trees, each of size at most $s$, such that for every input $x \in \{0,1\}^n$, sampling a decision tree $T$ from $\mathcal{D}$ and evaluating it on $x$ yields:
\[
\Pr_{T \sim \mathcal{D}}[T(x) = f(x)] \geq 2/3.
\]
By a standard Chernoff bound argument, we can assume $\mathcal{D}$ is supported on $O(n)$ trees, as such a distribution always exists. Each decision tree in the support of $\mathcal{D}$ can be converted into a generalized polynomial of sparsity and \( \ell_1 \)-norm at most \( s \) that agrees with the tree.  Taking a convex combination of these polynomials (weighted by \( \mathcal{D} \)) gives a generalized polynomial that \( 1/3 \)-approximates \( f \) with sparsity \( O(sn) \) and \( \ell_1 \)-norm \( \leq s \), implying $\agspar(f) = O(s n)$ and $\gdmnormapprox{f} \leq s$. Taking logarithms gives the desired bound (1) of (a) and (b).

For (2), we apply \cref{thm:mon-gen-sparsity} together with \cref{remark:gspar-dsize}, obtaining:
\[
\log \DSize(f) = O(\log^2 M(f) \cdot \log n) = O\left((\log \agspar(f))^4 \cdot (\log n)^3\right),
\]
\[
\log \DSize(f) = O(\log^2 M(f) \cdot \log n) = O\left((\log \gdmnormapprox{f})^4 \cdot (\log n)^3\right).\qedhere
\]
\end{proof}

Thus, for monotone functions, the complexity measures \( \gspar(f) \), \( \agspar(f) \), \( \gdmnorm{f} \), \( \gdmnormapprox{f} \), \( \DSize(f) \), and \( \RSize(f) \) are all polynomially related on the logarithmic scale, up to polylogarithmic factors in \( n \). In contrast, such a relationship fails for general functions; see \cref{rmk:sinkgspar}. Specifically, there exists a function \( f \) on \( n \) bits with \( \gspar(f) = O(\sqrt{n}) \) but \( \RSize(f) = 2^{\Omega(\sqrt{n})} \).

\bibliographystyle{plain}
\bibliography{BoolFnrefs}

\begin{thebibliography}{10}

\bibitem{aaronson2021degree}
Scott Aaronson, Shalev Ben-David, Robin Kothari, Shravas Rao, and Avishay Tal.
\newblock Degree vs. approximate degree and quantum implications of huang’s sensitivity theorem.
\newblock In {\em Proceedings of the 53rd Annual ACM SIGACT Symposium on Theory of Computing}, pages 1330--1342, 2021.

\bibitem{AaronsonS04}
Scott Aaronson and Yaoyun Shi.
\newblock Quantum lower bounds for the collision and the element distinctness problems.
\newblock {\em J. {ACM}}, 51(4):595--605, 2004.

\bibitem{AFH12}
Anil Ada, Omar Fawzi, and Hamed Hatami.
\newblock Spectral norm of symmetric functions.
\newblock In Anupam Gupta, Klaus Jansen, Jos{\'{e}} D.~P. Rolim, and Rocco~A. Servedio, editors, {\em Approximation, Randomization, and Combinatorial Optimization. Algorithms and Techniques - 15th International Workshop, {APPROX} 2012, and 16th International Workshop, {RANDOM} 2012, Cambridge, MA, USA, August 15-17, 2012. Proceedings}, volume 7408 of {\em Lecture Notes in Computer Science}, pages 338--349. Springer, 2012.

\bibitem{ACW25}
Josh Alman, Arkadev Chattopadhyay, and Ryan Williams.
\newblock Sparsity lower bounds for probabilistic polynomials.
\newblock In Raghu Meka, editor, {\em 16th Innovations in Theoretical Computer Science Conference, {ITCS} 2025, January 7-10, 2025, Columbia University, New York, NY, {USA}}, volume 325 of {\em LIPIcs}, pages 3:1--3:25. Schloss Dagstuhl - Leibniz-Zentrum f{\"{u}}r Informatik, 2025.

\bibitem{BealsBCMW01}
Robert Beals, Harry Buhrman, Richard Cleve, Michele Mosca, and Ronald de~Wolf.
\newblock Quantum lower bounds by polynomials.
\newblock {\em J. {ACM}}, 48(4):778--797, 2001.

\bibitem{BeameH09}
Paul Beame and Dang{-}Trinh Huynh{-}Ngoc.
\newblock Multiparty communication complexity and threshold circuit size of ac{\^{}}0.
\newblock In {\em 50th Annual {IEEE} Symposium on Foundations of Computer Science, {FOCS} 2009, October 25-27, 2009, Atlanta, Georgia, {USA}}, pages 53--62. {IEEE} Computer Society, 2009.

\bibitem{BellantoniPU92}
Stephen~J. Bellantoni, Toniann Pitassi, and Alasdair Urquhart.
\newblock Approximation and small-depth frege proofs.
\newblock {\em {SIAM} J. Comput.}, 21(6):1161--1179, 1992.

\bibitem{BW01}
Eli Ben{-}Sasson and Avi Wigderson.
\newblock Short proofs are narrow - resolution made simple.
\newblock {\em J. {ACM}}, 48(2):149--169, 2001.

\bibitem{BN21}
Gal Beniamini and Noam Nisan.
\newblock Bipartite perfect matching as a real polynomial.
\newblock In {\em 53rd ACM Symposium on Theory of Computing (STOC)}, pages 1118--1131. ACM, 2021.

\bibitem{BIVW16}
Andrej Bogdanov, Yuval Ishai, Emanuele Viola, and Christopher Williamson.
\newblock Bounded indistinguishability and the complexity of recovering secrets.
\newblock In {\em International Cryptology Conference}, volume LNCS 9816, pages 593--618. Springer, 2016.

\bibitem{BMTW19}
Andrej Bogdanov, Nikhil~S. Mande, Justin Thaler, and Christopher Williamson.
\newblock Approximate degree, secret sharing and concentration phenomena.
\newblock In {\em Approximation, Randomization and Combinatorial Optimization. Algorithms and Techniques, APPROX/RANDOM}, volume LIPIcs, 145, pages 71:1--71:21. Schloss-Dagstuhl, 2019.

\bibitem{BS90}
Jehoshua Bruck and Roman Smolensky.
\newblock Polyomial threshold functions, $\text{AC}^0$ functions, and spectral norms (extended abstract).
\newblock In {\em 31st Annual Symposium on Foundations of Computer Science (FOCS)}, volume~II, pages 632--641. IEEE, 1990.

\bibitem{bruck1992polynomial}
Jehoshua Bruck and Roman Smolensky.
\newblock Polynomial threshold functions, ac\^{}0 functions, and spectral norms.
\newblock {\em SIAM Journal on Computing}, 21(1):33--42, 1992.

\bibitem{BuhrmanW01}
Harry Buhrman and Ronald de~Wolf.
\newblock Communication complexity lower bounds by polynomials.
\newblock In {\em Proceedings of the 16th Annual {IEEE} Conference on Computational Complexity, Chicago, Illinois, USA, June 18-21, 2001}, pages 120--130. {IEEE} Computer Society, 2001.

\bibitem{BuhrmanW02}
Harry Buhrman and Ronald de~Wolf.
\newblock Complexity measures and decision tree complexity: a survey.
\newblock {\em Theor. Comput. Sci.}, 288(1):21--43, 2002.

\bibitem{BunKT20}
Mark Bun, Robin Kothari, and Justin Thaler.
\newblock The polynomial method strikes back: Tight quantum query bounds via dual polynomials.
\newblock {\em Theory Comput.}, 16:1--71, 2020.

\bibitem{BunT22}
Mark Bun and Justin Thaler.
\newblock Approximate degree in classical and quantum computing.
\newblock {\em Found. Trends Theor. Comput. Sci.}, 15(3-4):229--423, 2022.

\bibitem{chandrasekaran2014faster}
Karthekeyan Chandrasekaran, Justin Thaler, Jonathan Ullman, and Andrew Wan.
\newblock Faster private release of marginals on small databases.
\newblock In {\em Proceedings of the 5th conference on Innovations in theoretical computer science}, pages 387--402, 2014.

\bibitem{ChattopadhyayA08}
Arkadev Chattopadhyay and Anil Ada.
\newblock Multiparty communication complexity of disjointness.
\newblock {\em Electron. Colloquium Comput. Complex.}, {TR08-002}, 2008.

\bibitem{chattopadhyay2023randomized}
Arkadev Chattopadhyay, Yogesh Dahiya, Nikhil~S Mande, Jaikumar Radhakrishnan, and Swagato Sanyal.
\newblock Randomized versus deterministic decision tree size.
\newblock In {\em Proceedings of the 55th Annual ACM Symposium on Theory of Computing}, pages 867--880, 2023.

\bibitem{CM17eccc}
Arkadev Chattopadhyay and Nikhil~S. Mande.
\newblock Dual polynomials and communication complexity of {XOR} functions.
\newblock {\em Electron. Colloquium Comput. Complex.}, {TR17-062}, 2017.

\bibitem{CM17}
Arkadev Chattopadhyay and Nikhil~S. Mande.
\newblock A lifting theorem with applications to symmetric functions.
\newblock In Satya~V. Lokam and R.~Ramanujam, editors, {\em 37th {IARCS} Annual Conference on Foundations of Software Technology and Theoretical Computer Science, {FSTTCS} 2017, December 11-15, 2017, Kanpur, India}, volume~93 of {\em LIPIcs}, pages 23:1--23:14. Schloss Dagstuhl - Leibniz-Zentrum f{\"{u}}r Informatik, 2017.

\bibitem{CMS20}
Arkadev Chattopadhyay, Nikhil~S. Mande, and Suhail Sherif.
\newblock The log-approximate-rank conjecture is false.
\newblock {\em J.ACM}, 67(4):1--28, 2020.

\bibitem{cheung2025lower}
Tsun-Ming Cheung, Hamed Hatami, Kaave Hosseini, Aleksandar Nikolov, Toniann Pitassi, and Morgan Shirley.
\newblock A lower bound on the trace norm of boolean matrices and its applications.
\newblock In {\em 16th Innovations in Theoretical Computer Science Conference (ITCS 2025)}, pages 37--1. Schloss Dagstuhl--Leibniz-Zentrum f{\"u}r Informatik, 2025.

\bibitem{diakonikolas2010bounded}
Ilias Diakonikolas, Parikshit Gopalan, Ragesh Jaiswal, Rocco~A Servedio, and Emanuele Viola.
\newblock Bounded independence fools halfspaces.
\newblock {\em SIAM Journal on Computing}, 39(8):3441--3462, 2010.

\bibitem{EH89}
Andrzej Ehrenfeucht and David Haussler.
\newblock Learning decision trees from random examples.
\newblock {\em Information and Computation}, 82(3):231 -- 246, 1989.
\newblock Earlier version in COLT'88.

\bibitem{FurstSS84}
Merrick~L. Furst, James~B. Saxe, and Michael Sipser.
\newblock Parity, circuits, and the polynomial-time hierarchy.
\newblock {\em Math. Syst. Theory}, 17(1):13--27, 1984.

\bibitem{GS08}
Ben Green and Tom Sanders.
\newblock Boolean functions with small spectral norm.
\newblock {\em Geometric and Functional Analysis}, 18:144 -- 162, 2008.

\bibitem{grolmusz1997power}
Vince Grolmusz.
\newblock On the power of circuits with gates of low l1 norms.
\newblock {\em Theoretical computer science}, 188(1-2):117--128, 1997.

\bibitem{Hastad86}
Johan H{\aa}stad.
\newblock Almost optimal lower bounds for small depth circuits.
\newblock In Juris Hartmanis, editor, {\em Proceedings of the 18th Annual {ACM} Symposium on Theory of Computing, May 28-30, 1986, Berkeley, California, {USA}}, pages 6--20. {ACM}, 1986.

\bibitem{Hastad23}
Johan H{\aa}stad.
\newblock On small-depth frege proofs for {PHP}.
\newblock In {\em 64th {IEEE} Annual Symposium on Foundations of Computer Science, {FOCS} 2023, Santa Cruz, CA, USA, November 6-9, 2023}, pages 37--49. {IEEE}, 2023.

\bibitem{HastadRST17}
Johan H{\aa}stad, Benjamin Rossman, Rocco~A. Servedio, and Li{-}Yang Tan.
\newblock An average-case depth hierarchy theorem for boolean circuits.
\newblock {\em J. {ACM}}, 64(5):35:1--35:27, 2017.

\bibitem{Huang-19}
Hao Huang.
\newblock Induced subgraphs of hypercubes and a proof of the sensitivity conjecture.
\newblock {\em CoRR}, abs/1907.00847, 2019.

\bibitem{IPS99}
Russell Impagliazzo, Pavel Pudl{\'{a}}k, and Jir{\'{\i}} Sgall.
\newblock Lower bounds for the polynomial calculus and the gr{\"{o}}bner basis algorithm.
\newblock {\em Comput. Complex.}, 8(2):127--144, 1999.

\bibitem{Jukna-BFCbook}
Stasys Jukna.
\newblock {\em Boolean Function Complexity - Advances and Frontiers}, volume~27 of {\em Algorithms and Combinatorics}.
\newblock Springer, 2012.

\bibitem{kalai2008agnostically}
Adam~Tauman Kalai, Adam~R Klivans, Yishay Mansour, and Rocco~A Servedio.
\newblock Agnostically learning halfspaces.
\newblock {\em SIAM Journal on Computing}, 37(6):1777--1805, 2008.

\bibitem{Klauck07}
Hartmut Klauck.
\newblock Lower bounds for quantum communication complexity.
\newblock {\em {SIAM} J. Comput.}, 37(1):20--46, 2007.

\bibitem{klivans2004learning}
Adam~R Klivans and Rocco~A Servedio.
\newblock Learning dnf in time 2{\~o} (n1/3).
\newblock {\em Journal of Computer and System Sciences}, 68(2):303--318, 2004.

\bibitem{knop2021guest}
Alexander Knop, Shachar Lovett, Sam McGuire, and Weiqiang Yuan.
\newblock Guest column: Models of computation between decision trees and communication.
\newblock {\em ACM SIGACT News}, 52(2):46--70, 2021.

\bibitem{knop2021log}
Alexander Knop, Shachar Lovett, Sam McGuire, and Weiqiang Yuan.
\newblock Log-rank and lifting for and-functions.
\newblock In {\em Proceedings of the 53rd Annual ACM SIGACT Symposium on Theory of Computing}, pages 197--208, 2021.

\bibitem{KrajicekPW95}
Jan Kraj{\'{\i}}cek, Pavel Pudl{\'{a}}k, and Alan~R. Woods.
\newblock An exponenetioal lower bound to the size of bounded depth frege proofs of the pigeonhole principle.
\newblock {\em Random Struct. Algorithms}, 7(1):15--40, 1995.

\bibitem{LeeS09}
Troy Lee and Adi Shraibman.
\newblock Disjointness is hard in the multiparty number-on-the-forehead model.
\newblock {\em Comput. Complex.}, 18(2):309--336, 2009.

\bibitem{Nisan21}
Noam Nisan.
\newblock The demand query model for bipartite matching.
\newblock In {\em Proceedings of the 2021 ACM-SIAM Symposium on Discrete Algorithms (SODA)}, pages 592--599. ACM, 2021.

\bibitem{nisan1994degree}
Noam Nisan and Mario Szegedy.
\newblock On the degree of boolean functions as real polynomials.
\newblock {\em Computational complexity}, 4:301--313, 1994.

\bibitem{PitassiBI93}
Toniann Pitassi, Paul Beame, and Russell Impagliazzo.
\newblock Exponential lower bounds for the pigeonhole principle.
\newblock {\em Comput. Complex.}, 3:97--140, 1993.

\bibitem{R03}
Alexander Razborov.
\newblock Quantum communication complexity of symmetric predicates.
\newblock {\em Izvestiya:Mathematics}, 67(1):145--159, 2003.

\bibitem{Sherstov10}
Alexander~A. Sherstov.
\newblock On quantum-classical equivalence for composed communication problems.
\newblock {\em Quantum Inf. Comput.}, 10(5{\&}6):435--455, 2010.

\bibitem{Sherstov11}
Alexander~A. Sherstov.
\newblock The pattern matrix method.
\newblock {\em {SIAM} J. Comput.}, 40(6):1969--2000, 2011.

\bibitem{Sherstov14}
Alexander~A. Sherstov.
\newblock Communication lower bounds using directional derivatives.
\newblock {\em J. {ACM}}, 61(6):34:1--34:71, 2014.

\bibitem{SZ09}
Yaoyun Shi and Yufan Zhu.
\newblock Quantum communication complexity of block-composed functions.
\newblock {\em Quantum Inf. Comput.}, 9(5{\&}6):444--460, 2009.

\bibitem{thaler2012faster}
Justin Thaler, Jonathan Ullman, and Salil Vadhan.
\newblock Faster algorithms for privately releasing marginals.
\newblock In {\em International Colloquium on Automata, Languages, and Programming}, pages 810--821. Springer, 2012.

\bibitem{zhang2014efficient}
Shengyu Zhang.
\newblock Efficient quantum protocols for xor functions.
\newblock In {\em Proceedings of the twenty-fifth annual ACM-SIAM symposium on Discrete algorithms}, pages 1878--1885. SIAM, 2014.

\end{thebibliography}
\end{document}